\documentclass[a4paper,11pt]{article}
\usepackage{a4wide}

\usepackage[UKenglish]{babel}
\usepackage[noadjust]{cite}
\usepackage{amsmath}
\usepackage{amssymb}
\usepackage{amsthm} 
\usepackage{bm}
\usepackage{bbm}
\usepackage{mathtools}
\usepackage{mathrsfs}
\usepackage{csquotes}
\usepackage[ruled,linesnumbered,lined,boxed]{algorithm2e}

\theoremstyle{plain}
\newtheorem{thm}{Theorem}
\newtheorem*{thm*}{Theorem}
\newtheorem{lem}[thm]{Lemma}
\newtheorem*{lem*}{Lemma}
\newtheorem{prop}[thm]{Proposition}
\theoremstyle{definition}
\newtheorem{defn}[thm]{Definition}
\newtheorem{rem}[thm]{Remark}
\newtheorem{example}[thm]{Example}

\usepackage{hyperref} 
\hypersetup{colorlinks=true,citecolor=blue,linkcolor=blue,urlcolor=blue}

\newcommand{\A}{\mathbf{A}}
\newcommand{\B}{\mathbf{B}}
\newcommand{\C}{\mathbf{C}}
\newcommand{\K}{\mathbf{K}}
\newcommand{\X}{\mathbf{X}}

\newcommand{\N}{\mathbb{N}}
\newcommand{\R}{\mathbb{R}}
\newcommand{\Q}{\mathbb{Q}}
\newcommand{\Z}{\mathbb{Z}}

\renewcommand{\vec}[1]{\mathbf{#1}}
\newcommand{\ba}{\vec{a}}
\newcommand{\bb}{\vec{b}}
\newcommand{\bc}{\vec{c}}
\newcommand{\br}{\vec{r}}
\newcommand{\bu}{\vec{u}}
\newcommand{\bx}{\vec{x}}
\newcommand{\bv}{\vec{v}}

\newcommand{\bw}{\vec{w}}
\newcommand{\be}{\vec{e}}
\newcommand{\bq}{\vec{q}}
\newcommand{\bz}{\vec{z}}

\newcommand{\btau}{\bm{\tau}}
\newcommand{\bmu}{\bm{\mu}} 
\newcommand{\bdelta}{\bm{\delta}} 
\newcommand{\brho}{\bm{\rho}} 
\newcommand{\bxi}{\bm{\xi}}
\newcommand{\bzeta}{\bm{\zeta}}

\DeclareMathOperator{\BLP}{BLP}
\DeclareMathOperator{\CBLP}{CBLP}
\DeclareMathOperator{\SBLP}{SBLP}
\DeclareMathOperator{\AIP}{AIP}
\DeclareMathOperator{\SAC}{SAC}
\DeclareMathOperator{\Pol}{Pol}
\DeclareMathOperator{\PCSP}{PCSP}
\DeclareMathOperator{\CSP}{CSP}
\DeclareMathOperator{\ar}{ar}
\DeclareMathOperator{\CLAP}{CLAP}
\DeclareMathOperator{\power}{\mathcal{P}}
\DeclareMathOperator{\maj}{maj}

\DeclarePairedDelimiter{\floor}{\lfloor}{\rfloor}

\newcommand{\Qconv}{\ensuremath{\mathscr{Q}_{\operatorname{conv}}}}
\newcommand{\Zaff}{\ensuremath{\mathscr{Z}_{\operatorname{aff}}}}
\newcommand{\Mblpaip}{\ensuremath{\mathscr{M}_{\operatorname{\BLP+\AIP}}}}
\newcommand{\ourMinion}{\mathscr{C}}

\newcommand{\bone}{\mathbf{1}}  
\newcommand{\bzero}{\mathbf{0}} 

\DeclareMathOperator{\diag}{diag}
\DeclareMathOperator{\supp}{supp}

\newcommand{\inn}[2]{\mathbf{#1}\textbf{-in-}\mathbf{#2}}
\newcommand{\NAE}{\mathbf{NAE}}
\newcommand{\red}{\leq_p}
\newcommand{\un}{\ensuremath{R_{\operatorname{u}}}}

\begin{document}

\title{CLAP: A New Algorithm for Promise CSPs\thanks{An extended abstract of this work appeared in the Proceedings of the 2022 ACM-SIAM Symposium on Discrete Algorithms (SODA'22)~\cite{Ciardo22:soda}. The research leading to these results has received funding from the European Research Council (ERC) under the European Union's Horizon 2020 research and innovation programme (grant agreement No 714532). Stanislav \v{Z}ivn\'y was supported by a Royal Society University Research Fellowship. The paper reflects only the authors' views and not the views of the ERC or the European Commission. The European Union is not liable for any use that may be made of the information contained therein. This work was also supported by UKRI EP/X024431/1.
For the purpose of Open Access, the authors have applied a CC BY public copyright licence to any Author Accepted Manuscript version arising from this submission. All data is provided in full in the results section of this paper.}}
\author{Lorenzo Ciardo\thanks{Department of Computer Science, University of Oxford, UK.} \and Stanislav {\v{Z}}ivn{\'y}$^\dagger$}
	
\date{\today}
\maketitle

\begin{abstract} 
  We propose a new algorithm for Promise Constraint Satisfaction Problems
  (PCSPs). It is a combination of the \textbf{C}onstraint Basic \textbf{L}P
  relaxation and the \textbf{A}ffine I\textbf{P} relaxation ($\CLAP$). We give a
  characterisation of the power of $\CLAP$ in terms of a minion homomorphism. Using this
  characterisation, we identify a certain weak notion of symmetry which, if
  satisfied by infinitely many polymorphisms of PCSPs, guarantees tractability.
  
  We demonstrate that there are PCSPs solved by $\CLAP$ that are not solved by
  any of the existing algorithms for PCSPs; in particular, not by the $\BLP+\AIP$
  algorithm of Brakensiek et al.~[SICOMP'20] and not by a reduction to tractable
  finite-domain CSPs.
\end{abstract}	

\section{Introduction}
\label{sec:intro}

\paragraph{Constraint Satisfaction}
Constraint Satisfaction Problems (CSPs) have driven some of the most influential
developments in theoretical computer science, from NP-completeness to the PCP
theorem~\cite{Arora98:jacm-probabilistic,Arora98:jacm-proof,Dinur07:jacm} to
semidefinite programming algorithms~\cite{Raghavendra08:everycsp} to the Unique
Games Conjecture~\cite{Khot02stoc}.

A CSP over domain $A$ is specified by a finite collection $\A$ of relations over
$A$, and is denoted by $\CSP(\A)$. Given on input a set of variables and a set
of constraints, each of which uses relations from $\A$, the task is to decide
the existence of an assignment of values from $A$ to the variables that
satisfies
all the constraints. Classic examples of CSPs include $2$-SAT, graph
$3$-colouring, and linear equations of fixed width over finite groups.

For Boolean CSPs, which are CSPs with
$|A|=2$, Schaefer proved 
that every such CSP is either solvable in polynomial time or is
NP-complete~\cite{Schaefer78:stoc}. Feder and Vardi famously conjectured that
the same holds true for CSPs over arbitrary finite
domains~\cite{Feder98:monotone}. Furthermore, they realised the
importance of considering closure properties of solution spaces of CSPs~\cite{Feder98:monotone}, which initiated the algebraic
approach~\cite{Jeavons97:closure,Jeavons98:algebraic,Bulatov05:classifying}. The key
notion in the algebraic approach is that of \emph{polymorphisms}, which are operations that take
solutions to a CSP and are guaranteed to return, by a coordinatewise
application, a solution to the same CSP. All CSPs admit projections (also known
as dictators) as polymorphisms. However, the presence of less trivial polymorphisms, satisfying some notion of symmetry, is necessary for tractability. For instance, the set of solutions to
$2$-SAT is closed under the ternary majority operation $\maj:\{0,1\}^3\to\{0,1\}$
that satisfies the following notion of symmetry: $\maj(a,a,b)=\maj(a,b,a)=\maj(b,a,a)=a$ for
any $a,b\in\{0,1\}$. Similarly, the set of solutions to Horn-SAT is closed under
the binary minimum operation $\min:\{0,1\}^2\to\{0,1\}$ that satisfies a different
notion of symmetry: $\min(a,a)=a$, $\min(a,b)=\min(b,a)$, and
$\min(a,\min(b,c))=\min(\min(a,b),c)$
for any $a,b,c\in\{0,1\}$. The binary
$\max$ operation -- which is a polymorphism of dual Horn-SAT -- has the same notion of symmetry, called
semilattice~\cite{BKW17}. Together with the ternary minority polymorphism,
which captures linear equations on $\{0,1\}$, this gives all non-trivial tractable cases
from Schaefer's dichotomy result.\footnote{The trivial cases, called $0$-
and $1$-valid, are captured by the constant-$0$ and constant-$1$
polymorphisms, respectively.}

The polymorphisms of any CSP form a clone, in that they include all projections
and are closed under
composition. For instance, since Horn-SAT has $\min$ as a polymorphism, it also
has the $4$-ary minimum operation $\min_4(a,b,c,d)=\min(a,\min(b,\min(c,d)))$
as a polymorphism. Building on the connection to universal algebra,
the algebraic approach has been tremendously successful beyond decision CSPs,
e.g. for robust
satisfiability of CSPs~\cite{Dalmau13:toct,Barto16:sicomp,Dalmau19:sicomp}, for exact
optimisation of CSPs~\cite{Kozik15:icalp,tz16:jacm-complexity,Kolmogorov17:sicomp}, and for
characterising the power of
algorithms~\cite{Kun12:itcs,Barto14:jacm,Berman10:few,ktz15:sicomp,Kozik21:sicomp,tz17:sicomp,tz18}. The
culmination of the algebraic approach is the positive resolution of the
dichotomy conjecture by Bulatov~\cite{Bulatov17:focs} and
Zhuk~\cite{Zhuk20:jacm}. We refer the reader to~\cite{BKW17} for a
survey on the algebraic approach.

\paragraph{Promise Constraint Satisfaction}
In this paper, we study Promise Constraint Satisfaction Problems (PCSPs), 
whose systematic study was initiated by Austrin, Guruswami, and
H{\aa}stad~\cite{AGH17}, and Brakensiek and Guruswami~\cite{BG21}. 
PCSPs form a vast generalisation of CSPs. In $\PCSP(\A,\B)$, each constraint comes in two
forms, a strict one in $\A$ and a weak one in $\B$. The goal is to distinguish between (i) the
case in which (the strong form of)  the constraints can be simultaneously satisfied
in $\A$ and (ii) the case in which (even the weak form of) the constraints cannot be
simultaneously satisfied in $\B$. The promise is that it is never the case that the PCSP
is not satisfiable in the strict sense but is satisfiable in the weak sense.
If the strict and weak forms coincide in every constraint (i.e., if $\A=\B$) we get the (non-promise)
CSPs. However, PCSPs include many fundamental problems that are inexpressible
as CSPs. 

The simplest example of strict vs. weak constraints is when the weak constraints
are supersets of the strict constraints on the same domain (the first two
examples below) or on a larger domain (the third example below); the notion of
homomorphism from $\A$ to $\B$ formalises this for any PCSP.

First, can we distinguish a $g$-satisfiable $k$-SAT instance (in the
sense that there is an assignment that satisfies at least $g$ literals in each
clause) from an instance that is not even $1$-satisfiable? This problem was studied in~\cite{AGH17}, where it was shown to be solvable in polynomial time if $\frac{g}{k}\geq\frac{1}{2}$ and
NP-complete otherwise. Recently, this result has been generalised to arbitrary
finite domains~\cite{BWZ21}.

Second, can we distinguish a $3$-SAT
formula that admits an assignment satisfying exactly $1$ literal in each
clause (i.e., a satisfiable instance of $\inn{1}{3}$-SAT) from one that does not admit an assignment satisfying $1$ or $2$ literals in each
clause (i.e., a non-satisfiable instance of \textbf{N}ot-\textbf{A}ll-\textbf{E}qual-$3$-SAT)? Remarkably, while both $\inn{1}{3}$
and $\NAE$ are NP-hard, this promise version 
is solvable in polynomial time~\cite{BG21,BG19}.

Third, can we distinguish a $k$-colourable graph from a graph that is not even $\ell$-colourable, where
$k\leq\ell$? This is the \emph{approximate graph colouring}
problem, which is believed to be NP-hard for any fixed $3\leq
k\leq\ell$, but has been elusive since the 1970s~\cite{GJ76}. In particular, the larger the gap is between $k$ and $\ell$ the easier the problem could in principle
be and, thus, the more challenging it is to prove NP-hardness. The current
state of the art is NP-hardness for $k=3$ and $\ell=5$~\cite{BBKO21}, while already the
case of $k=3$ and $\ell=6$ is open. For any
$k\geq 4$ and $\ell=\ell(k)=\binom{k}{\floor k/2}-1$, NP-hardness has been
established in~\cite{WZ20}.

\medskip

While a systematic study of PCSPs was initiated only recently~\cite{AGH17,BG21},
concrete PCSPs have been considered for a while, e.g. approximate
graph~\cite{GJ76,Wigderson83:jacm,Blum94:jacm-new,Khanna00:combinatorica,Khot01,Guruswami04:sidma}
and hypergraph colouring~\cite{DRS05}.
A highlight
result is the dichotomy of Boolean symmetric
PCSPs~\cite{Ficak19:icalp} (in which all constraint relations are symmetric), following an earlier classification of Boolean
symmetric PCSPs with disequalities~\cite{BG21}. Very recent works have
investigated certain Boolean non-symmetric PCSPs~\cite{bz22:ic} and certain
non-Boolean symmetric PCSPs~\cite{Barto21:stacs}. 
Other recent results include, e.g.,~\cite{ABP20,GS20:icalp,BGS21}. 

Most of the recent progress, including results on the approximate graph
colouring problem~\cite{BBKO21,WZ20} and on the approximate graph
homomorphism problem~\cite{KO19,WZ20}, rely on the algebraic approach to
PCSPs~\cite{BBKO21}. In particular, the breakthrough results
in~\cite{BBKO21}, building on~\cite{BOP18}, established that the
complexity of PCSPs is captured by 
the polymorphism minions and certain types of symmetries
these minions satisfy -- these are non-nested identities on polymorphisms, such
as the majority example but not the semilattice example. Crucially, minions are less structured than
clones: A minion (of functions) is a set of operations closed under permuting
coordinates, identifying coordinates, and introducing dummy coordinates, but not under
composition.\footnote{In this work, we shall use the more abstract notion of
minion introduced in~\cite{bgwz20}, cf.~Definition~\ref{defn_minion}.} Thus,
unlike in our earlier CSP example (corresponding to Horn-SAT), a binary minimum polymorphism of a PCSP
cannot in general be used to generate a $4$-ary minimum polymorphism of the same PCSP.

Despite the momentous results in~\cite{BBKO21}, there is a long way to go to classify all
PCSPs, and it is not even clear whether a dichotomy for all PCSPs should be
expected. When Feder and Vardi conjectured a CSP
dichotomy~\cite{Feder98:monotone}, the Boolean
case~\cite{Schaefer78:stoc} and the graph case~\cite{HellN90} had been fully 
classified. We
seem quite far from these two cases being classified for PCSPs.
Thus, further progress is needed on both the hardness and tractability part.
This paper focusses on the latter.

\paragraph{Finite tractability}
Although PCSPs are (much) more general than CSPs, some PCSPs can be reduced to
tractable CSPs. This idea was introduced in~\cite{BG19} under the name of
homomorphic sandwiching (cf.~Section~\ref{sec:prelims} for a precise definition); PCSPs that are reducible to tractable (finite-domain) CSPs are called
\emph{finitely tractable}. Finite tractability is not sufficient to explain
tractability of all tractable PCSPs. In particular, Barto et al.~\cite{BBKO21} showed that
the above-mentioned example $\inn{1}{3}$ vs. $\NAE$ is not
finitely tractable, despite being a tractable PCSP~\cite{BG21}. We remark that
it is not inconceivable (and in fact was conjectured in~\cite{BG19}) that every
tractable (finite-domain) PCSP could be reducible to a tractable CSP possibly
over an infinite domain; this is the case for the $\inn{1}{3}$ vs. $\NAE$ problem~\cite{BG19}. However, while certain infinite-domain CSPs
are amenable to algebraic methods, the complexity of infinite-domain CSPs is far
from understood, cf.~\cite{BMM18,Bodirsky21,Barto20:sicomp} for recent
work.

Since finite tractability does not capture all tractable PCSPs, there is need
for other algorithmic tools. One possibility is to attempt to extend algorithmic techniques developed for CSPs.

There are two main algorithmic approaches to CSPs. On the one hand, there are
local consistency methods~\cite{Feder98:monotone}, which have been 
studied in
theoretical computer science but also in artificial intelligence, logic, and database theory. The power of local
consistency for CSPs has been characterised
in~\cite{Bulatov09:width,Barto14:jacm}, and it is known that the third level of
consistency solves all so-called bounded-width CSPs~\cite{Barto14:jloc}. On the
other hand, there are CSPs solvable by algorithms based on
generalisations of Gaussian elimination, most notably CSPs with a Mal'tsev
polymorphism~\cite{Bulatov06:sicomp}. 
This method has been pushed to its limit, in a way,
in~\cite{Idziak10:siam,Berman10:few}.
While the NP-hardness part of the CSP dichotomy has been known since~\cite{Bulatov05:classifying}, the challenge in proving the algorithmic part
is the complicated interaction of these two very different
algorithmic approaches.
Although this interaction does not occur in Boolean CSPs, it occurs already in CSPs
on three-element domains~\cite{Bulatov06:jacm}.

The characterisation of the power of the first level of the consistency methods,
$1$-consistency (also known as arc-consistency~\cite{Mackowrth77:aij}), has been
lifted from CSPs~\cite{Feder98:monotone} to PCSPs in~\cite{BBKO21}. Rather than
establishing $1$-consistency combinatorially, one can employ convex relaxations.

\paragraph{Relaxations}
A canonical analogue of $1$-consistency is the \emph{basic linear programming
relaxation} ($\BLP$)~\cite{Kun12:itcs}, which in fact is stronger than
$1$-consistency~\cite{Kun16}. The characterisation of the power of $\BLP$ has been
lifted from CSPs~\cite{Kun12:itcs} to PCSPs in~\cite{BBKO21}, both in terms of a
minion and a property of polymorphisms. The power of $\BLP$ is captured by a
minion consisting of rational stochastic vectors\footnote{A vector is \emph{stochastic} if its entries are nonnegative and sum up to one.} or, equivalently, by the
presence of symmetric polymorphisms of all arities; these are polymorphisms
invariant under any permutation of the coordinates. For example, we have seen
that Horn-SAT, a classic CSP, has a binary symmetric polymorphism, namely $\min$. We have also
seen that $\min$ can generate a $4$-ary operation $\min_4$, which is symmetric.
Similarly, $\min$ can generate (via composition) symmetric operations of all
arities, and thus Horn-SAT is solved by $\BLP$. 

A different relaxation of PCSPs is the \emph{basic affine integer programming
relaxation} ($\AIP$)~\cite{BG19}. The power of $\AIP$ has been characterised, both in
terms of a minion and a property of polymorphisms, in~\cite{BBKO21}. The minion
capturing $\AIP$ consists of integer affine vectors.\footnote{An integer vector
is \emph{affine} if its entries sum up to one.} Concerning polymorphisms, $\AIP$
is captured by polymorphisms of all odd arities that are invariant under
permutations that only permute odd and even coordinates separately, and
additionally satisfy that adjacent coordinates cancel each other out. The
$\inn{1}{3}$ vs. $\NAE$ problem is solved by $\AIP$ (cf.~Example~\ref{ex:aip}).

Brakensiek et al.~\cite{bgwz20} proposed a combination of the two above-mentioned relaxations, called $\BLP+\AIP$. Their algorithm has many interesting features. Firstly, it solves PCSPs that admit only infinitely many
symmetric polymorphisms (i.e., not all arities are required as in the case of
$\BLP$). Secondly, it solves all tractable Boolean CSPs, thus demonstrating how
research on PCSPs can shed new light on (non-promise) CSPs. 
In fact, \cite{bgwz20} established the power of
$\BLP+\AIP$ in terms of a minion and (a property of) polymorphisms. The minion
capturing $\BLP+\AIP$ is essentially a product of the $\BLP$ and $\AIP$
minions~\cite{bgwz20}. Concerning polymorphisms, $\BLP+\AIP$ is captured by
polymorphisms of all odd arities that are invariant under permutations that only
permute odd and even coordinates. 

It may be that $\BLP+\AIP$ is the only algorithm needed to solve all tractable \emph{Boolean} PCSPs. 
However, as already observed in~\cite{bgwz20}, $\BLP+\AIP$ does not solve some
rather simple, tractable, \emph{non-Boolean} PCSPs. 
Motivated by this, we investigate algorithms that are stronger than $\BLP+\AIP$.
We note that all PCSPs hitherto known to be tractable are solved by $\BLP+\AIP$ or by finite tractability (i.e., by a reduction to a tractable finite-domain CSP). In this work, we provide an example of a PCSP that is tractable (through our algorithm) but is not solved by either of those two algorithmic techniques.

\paragraph{Contributions} 

Building on the work of Brakensiek et al.~\cite{bgwz20}, we study stronger
relaxations for PCSPs and give three main contributions.

\smallskip\noindent \textbf{(1) CLAP} \quad
Our first contribution is the introduction of $\CLAP$ to the study of PCSPs. 
Our goal was to design an algorithm that, unlike
$\BLP+\AIP$, solves all CSPs of bounded width.
While all bounded-width CSPs can be solved by
$3$-consistency~\cite{Barto14:jloc}, and thus also by the third level of the
Sherali-Adams hierarchy for $\BLP$ (e.g., by~\cite{tz17:sicomp}), Kozik showed
that already 
(a special case of) the singleton arc-consistency ($\SAC$) algorithm,
introduced in~\cite{DB97}
(cf.~\cite{BD08,Chen13:jlc}), 
solves all bounded-width
CSPs~\cite{Kozik21:sicomp}.
Thus, we study the LP relaxation that we call the \emph{singleton} $\BLP$ ($\SBLP$),
which is at least as strong as $\SAC$. A special case of $\SBLP$ (without this name)
implicitly appeared in the literature, e.g. in~\cite{AGH17,BG21} for Boolean
PCSPs.
The idea behind $\SBLP$ is essentially to run $\SAC$ but replace the arc-consistency
check by the $\BLP$; i.e., the algorithm repeatedly takes a variable-value pair
$(x,a)$ and tests the feasibility of the $\BLP$ with the requirement that $x$
should be assigned the value $a$. If this LP is infeasible then $a$ is removed
from the domain of $x$. This is repeated until convergence. If any variable ends
up with an empty domain then $\SBLP$ rejects, otherwise it accepts. Overall, the
number of $\BLP$ calls occurring for an instance of $\PCSP(\A,\B)$ with
variable-set $X$ is at most polynomial in the size of $X$. As mentioned above,
this simple algorithm solves all bounded-width CSPs~\cite{Kozik21:sicomp}.

We adopt a modification of $\SBLP$ that turns out to be more naturally captured
by a minion-oriented analysis: the \emph{constraint} $\BLP$ ($\CBLP$). This
(possibly) stronger algorithm is a generalisation of $\SBLP$ in which
we do not consider only variable-value pairs $(x,a)$, but rather the
constraint-assignment pairs $(\bx,\ba)$ for every constraint in the instance. As
in $\SBLP$, if fixing a (local) assignment to a constraint yields an infeasible
$\BLP$ then the assignment is removed from the constraint relation. Upon
convergence, which takes at most polynomially many $\BLP$ calls, if any constraint
ends up with an empty relation then $\CBLP$ rejects, otherwise it accepts. 

Our algorithm $\CLAP$ first runs $\CBLP$ and then, upon termination, refines the solutions of $\CBLP$ by running (essentially) $\AIP$. If one believes the suggestion in~\cite{bgwz20} that constantly many
rounds of the Sherali-Adams hierarchy for $\BLP+\AIP$ could solve all tractable
(non-promise) CSPs, then it is not outrageous to believe that the same could be
true for $\CLAP$, and $\CLAP$ might be easier to analyse than such an algorithm.

\smallskip\noindent\textbf{(2) Characterisation}\quad
Our second contribution is a minion characterisation of the power of $\CLAP$,
stated as Theorem~\ref{thm:main1}. The objects in the minion are essentially
matrices with a particular structure, which we call skeletal (cf.~Definition~\ref{def:skeletal}). These matrices capture the $\CBLP$ part of
$\CLAP$ and together with certain integer affine vectors form the minion (cf.~Definition~\ref{def:minion}). Another, more conceptual contribution is the introduction of a minion of matrices to the study of PCSPs.

\smallskip\noindent\textbf{(3) $\textbf{H}$-symmetric polymorphisms}\quad
The minion characterisation is crucial to our third contribution: the identification of a sufficient condition for
$\CLAP$ to work in terms of the symmetries of the polymorphisms. This
is stated as Theorem~\ref{thm:Hsymm}, using the notion of $H$-symmetry. This condition can be more easily checked for concrete templates, thus allowing us to design a separating example that is
not finitely tractable and is
 not solved
by $\BLP+\AIP$ (nor by local consistency methods, see~\cite{Atserias22:soda}), but is solved by $\CLAP$. It follows that our new algorithm is strictly more powerful than $\BLP+\AIP$ (and separated by an interesting PCSP that is not reducible to a tractable finite-domain CSP via ``gadget reductions'', which capture the algebraic approach to PCSPs~\cite{BBKO21}).
 
For a matrix $H$, a polymorphism $f$ is $H$-symmetric if $f$ is
invariant under permutations of the coordinates but only on a specific set of inputs determined by $H$ (cf.~Definition~\ref{def_H_tieless_symmetric}). For instance, if $H$ is a row vector 
then we obtain the requirement that $f$ be symmetric on all inputs.
If $H$ is the identity matrix then we require that $f$ be symmetric only on
inputs in which different entries occur with different multiplicities. In
general, the intuition is that we capture \enquote{symmetry with 
exceptions that depend on multiplicities}. We refer the reader to the discussion in
Section~\ref{sec:clap} for details. 

After necessary background material in Section~\ref{sec:prelims}, our algorithm $\CLAP$ and the
main results are presented in Section~\ref{sec:clap}; the proofs appear in Sections~\ref{sec_power_CLAP_minions} and~\ref{sec_H_symmetric}.

\section{Preliminaries}
\label{sec:prelims}
We
let $\N=\{1,2,\ldots\}$ and $\N_0=\N\cup\{0\}$. The cardinality of $\N$ shall be denoted by $\aleph_0$.
For $k\in\N$, $[k]$ denotes the set $\{1,\ldots,k\}$. For a set
$A$, $\power(A)$ denotes the set of all subsets of $A$.
We denote by $\red$ many-one polynomial-time reductions.
We shall use standard notation for vectors and matrices. Vectors will be treated
as column vectors and whenever convenient identified with the corresponding
(row) tuples. Both tuples and vectors will be typed in bold font. We denote by
$\be_i$ the $i$-th standard unit vector of the appropriate size (which will
be clear from the context); i.e., $\be_i$ is equal to $1$ in the $i$-th
coordinate and $0$ elsewhere. We denote by $\bzero_p$ and by $\bone_p$ the
all-zero and all-one vector, respectively, of size $p$; if the size is clear, we occasionally drop the subscript. 
The \emph{support}
of a 
%real 
vector $\bv=(v_i)$ of size $p$ is the set $\supp(\bv)=\{i\in [p]:v_i\neq 0\}$. 
$I_p$ denotes the identity matrix of order $p$, while $O$ denotes an all-zero matrix of suitable size.

\paragraph{Promise CSPs}
A \emph{signature} $\sigma$ is a finite set of relation symbols $R$, each with
its arity $\ar(R)\in\N$. A \emph{relational structure} over a signature
$\sigma$, or a $\sigma$-structure, is a finite universe $A$, called
the \emph{domain} of $\A$, and a relation $R^\A\subseteq A^{\ar(R)}$ for
each symbol $R\in\sigma$. 
For two $\sigma$-structures $\A$ and $\B$, a mapping 
$f:A\to B$ is called a \emph{homomorphism} from $\A$ to $\B$, denoted by
$f:\A\to\B$, if $f$ preserves all relations; that is, for every $R\in\sigma$ and
every tuple $\textbf{a}\in R^\A$, we have $f(\textbf{a})\in R^\B$, where $f$ is applied
coordinatewise. The existence of a homomorphism from $\A$ to $\B$ is denoted by
$\A\to\B$. A PCSP \emph{template} is a pair $(\A,\B)$ of relational structures
over the same signature such that $\A\to\B$.
Without loss of generality, we will often assume that $A$, the domain of $\A$, is $[n]$.

\begin{defn}
  Let $(\A,\B)$ be a PCSP template. 
  The \emph{decision version} of $\PCSP(\A,\B)$ is the following problem: Given
  as input a relational structure $\X$ over the same signature as $\A$ and $\B$,
  output \textsc{Yes} if $\X\to\A$ and $\textsc{No}$ if $\X\not\to\B$.
  The \emph{search version} of $\PCSP(\A,\B)$ is the following problem: Given as
  input a relational structure $\X$ over the same signature as $\A$ and $\B$ and such
  that $\X\to\A$, find a homomorphism from $\X$ to $\B$.
\end{defn}

For a relational structure $\A$, the \emph{constraint satisfaction problem}
(CSP) with template $\A$~\cite{Feder98:monotone}, denoted by $\CSP(\A)$, is
$\PCSP(\A,\A)$.

\begin{example}
  For $k\geq 2$, let $\K_k$ be the structure with domain $[k]$
  and a binary relation $\{(i,j)\in [k]^2\mid i\neq j\}$. Then, $\CSP(\K_k)$ is
  the standard graph $k$-colouring problem. For $k\leq \ell$,
  $\PCSP(\K_k,\K_\ell)$ is the \emph{approximate graph colouring
  problem}~\cite{GJ76}. In the decision version, the task is to decide whether a graph
  is $k$-colourable or not even $\ell$-colourable. In the search version, given
  a $k$-colourable graph $G$, the task is to find an $\ell$-colouring of $G$. It
  is widely believed that for any fixed $3\leq k\leq \ell$,
  $\PCSP(\K_k,\K_\ell)$ is NP-hard; i.e., constantly many colours do not help.
  The current most general NP-hardness result is known for $k=3$ and $\ell=5$ by
  Bul\'in, Krokhin, and Opr\v{s}al~\cite{BBKO21} and for $k\geq 4$ and
  $\ell=\ell(k)=\binom{k}{\floor k/2}-1$ by
  Wrochna and \v{Z}ivn\'y~\cite{WZ20}.
\end{example}

We call a PCSP template $(\A,\B)$ \emph{tractable} if any instance of $\PCSP(\A,\B)$ can be
solved in polynomial time in the size of the input structure $\X$.
It is easy to show that the decision version reduces to the search
version~\cite{BBKO21} (but the converse is not known in general); for CSPs, the
two versions are equivalent~\cite{Cohen04,Bulatov05:classifying}.
Our results are for the decision version.

\begin{example}\label{ex:1-in-3}
  Let $\inn{1}{3}$ be the Boolean structure with domain $\{0,1\}$ and a single ternary
  relation $\{(0,0,1),(0,1,0),(1,0,0)\}$. Let $\NAE$ be the structure with
  domain $\{0,1\}$ and a single ternary relation
  $\{0,1\}^3\setminus\{(0,0,0),(1,1,1)\}$. Then, $\CSP(\inn{1}{3})$ is the
  (positive) $1$-in-$3$-SAT problem and $\CSP(\NAE)$ is the (positive)
  Not-All-Equal-$3$-SAT problem. Since both of these problems are
  NP-hard~\cite{Schaefer78:stoc}, the PCSP templates $(\inn{1}{3},\inn{1}{3})$
  and $(\NAE,\NAE)$ are both intractable. However, the PCSP template
  $(\inn{1}{3},\NAE)$ is tractable, as shown by Brakensiek and Guruswami~\cite{BG21}.
\end{example}

\begin{defn}
  Let $(\A,\B)$ be a PCSP template with signature $\sigma$. An operation $f:A^L\to
  B$, where $L\in\N$, is a \emph{polymorphism} of arity $L$ of $(\A,\B)$ if for every
  $R\in\sigma$ of arity $k=\ar(R)$ and for any possible $L\times k$ matrix
  whose rows are tuples in $R^\A$, the application of $f$ on the columns of the
  matrix gives a tuple in $R^\B$. We denote by $\Pol(\A,\B)$ the set of all
  polymorphisms of $(\A,\B)$.
\end{defn}

\begin{example}\label{ex:aip}
  The unary operation $\neg:\{0,1\}\to\{0,1\}$ defined by $\neg(a)=1-a$ is a
  polymorphism of $(\NAE,\NAE)$ but not a polymorphism of
  $(\inn{1}{3},\inn{1}{3})$. For any odd $L$, the $L$-ary operation
  $f:\{0,1\}^L\to\{0,1\}$ defined by
  $f(a_1,\ldots,a_L)=1$ if $a_1-a_2+a_3-\cdots+a_{L}>0$  and
  $f(a_1,\ldots,a_L)=0$ otherwise is a polymorphism of $(\inn{1}{3},\NAE)$.
\end{example}

\paragraph{Minions}
Polymorphisms of CSPs form clones; i.e., $\Pol(\A,\A)$ contains all
projections (also known as dictators) and is closed under
composition~\cite{BKW17}. Polymorphisms of the (more general) PCSPs form
\emph{minions}; i.e, they are closed under taking minors.\footnote{We remark that
clones are also closed under taking minors.} Formally, given an $L$-ary function $f:A^{L}\to B$, its \emph{minor} relative to a map $\pi:[L]\to[L']$ is the $L'$-ary function $f_{/\pi}:A^{L'}\to B$ defined by 
\begin{align}
\label{eq_minor_functions}
f_{/\pi}(a_1,\ldots,a_{L'})=f(a_{\pi(1)},\ldots,a_{\pi(L)}).
\end{align} 
Equivalently, a minor of $f$ is a function
obtained from $f$ by identifying variables, permuting variables, and introducing
dummy variables. Rather than focussing on minions of functions, we consider here
abstract minions, as described and used in~\cite{bgwz20}.
\begin{defn}
\label{defn_minion}
  A \emph{minion} $\mathscr{M}$ consists in the disjoint union of sets $\mathscr{M}^{(L)}$ for $L\in \N$ equipped with operations $(\cdot)_{/\pi}:\mathscr{M}^{(L)}\rightarrow\mathscr{M}^{(L')}$ for all functions $\pi:[L]\rightarrow [L']$, which satisfy
\begin{itemize}
\item $(M_{/\pi})_{/\tilde{\pi}}=M_{/\tilde{\pi}\circ \pi}$ for $\pi:[L]\rightarrow [L']$, $\tilde{\pi}:[L']\rightarrow [L'']$ and
\item $M_{/\operatorname{id}}=M$
\end{itemize}
for all $M\in\mathscr{M}^{(L)}$. 
\end{defn}
\begin{defn}
For two minions $\mathscr{M}$ and $\mathscr{N}$, a \emph{minion homomorphism} $\xi:\mathscr{M}\rightarrow\mathscr{N}$ is a map that preserves arities and minors: Given $M\in\mathscr{M}^{(L)}$ and $\pi:[L]\rightarrow[L']$, $\xi(M)\in \mathscr{N}^{(L)}$ and $\xi(M_{/\pi})=\xi(M)_{/\pi}$.
\end{defn}

For any PCSP template $(\A,\B)$, the set $\Pol(\A,\B)$ of its polymorphisms equipped with the operations described by~\eqref{eq_minor_functions}  is a
minion~\cite{BBKO21}.
One of the results in~\cite{BBKO21} established that minion homomorphisms
give rise to polynomial-time reductions: If there is a minion homomorphism from
$\Pol(\A,\B)$ to $\Pol(\A',\B')$, then $\PCSP(\A',\B')\red\PCSP(\A,\B)$. Minions
are also useful for characterising the power of algorithms, as we will discuss
later. 

\begin{rem}
Although we will not use this categorical view, we remark that a minion is nothing but a functor from the category of nonempty finite sets to the category of nonempty sets, and a minion homomorphism is a natural transformation.

\end{rem}

\paragraph{Existing algorithms}
One way to establish tractability of PCSPs is to reduce to CSPs.
Let $(\A,\B)$ be a PCSP template. A structure $\C$ is called a (homomorphic)
\emph{sandwich} if $\A\to\C\to\B$. It is known that, in this case,
$\PCSP(\A,\B)\red \CSP(\C)$.\footnote{This is a special case of homomorphic
relaxation~\cite{BBKO21}, which we do not need here.} Thus, if $\C$ is a
tractable CSP template then $(\A,\B)$ is a tractable PCSP template. If $\C$
has a finite domain,
we say that $(\A,\B)$ is \emph{finitely tractable}.

\begin{example}\label{ex:finite}
  The PCSP template $(\inn{1}{3},\NAE)$ from Example~\ref{ex:1-in-3} is
  tractable~\cite{BG21} but not finitely tractable unless P=NP, as shown in~\cite{BBKO21}.
\end{example}
Another way to establish tractability for PCSPs is to leverage convex
relaxations. In Section~\ref{sec:intro}, we mentioned three studied relaxations:
$\BLP$~\cite{Kun12:itcs}, $\AIP$~\cite{BG21}, and $\BLP+\AIP$~\cite{bgwz20}. Their powers have been characterised
in~\cite{BBKO21,bgwz20} in
terms of certain minions and polymorphism identities. The details of these relaxations and the
characterisations are provided in Appendix~\ref{app:relaxations}.

All PCSPs hitherto known to be tractable are solved by finite tractability (i.e., by a
reduction to a tractable finite-domain CSP) or by $\BLP+\AIP$.
 The next example
identifies a simple PCSP template not captured by either of these two methods.

\begin{example}
  \label{ex:new1}
  Consider the relational structures $\A=(A;R_1^\A,R_2^\A)$ and
  $\B=(B;R_1^\B,R_2^\B)$ on the domain $A=B=\{0,\ldots,6\}$ with the following
  relations: $R_1^\A=\{(0,0,1),(0,1,0),(1,0,0)\}$ is $\inn{1}{3}$ on $\{0,1\}$,
  $R_1^\B=\{0,1\}^3\setminus\{(0,0,0),(1,1,1)\}$ is $\NAE$ on $\{0,1\}$, and
  $R_2^\A=R_2^\B=\{(2,3),(3,2),(4,5),(5,6),(6,4)\}$.
  The identity mapping is a homomorphism from $\A$ to $\B$, so $(\A,\B)$ is a
  PCSP template. Since the directed graph corresponding to $R_2^\A=R_2^\B$ is a
  disjoint union of a directed $2$-cycle and a directed $3$-cycle,~\cite[Example
  6.1]{bgwz20} shows that the $\BLP+\AIP$ algorithm does not solve $\PCSP(\A,\B)$. 
  We claim that the template ($\A,\B$) is not finitely tractable. For
  contradiction, assume that there is a finite relational structure
  $\C=(C;R_1^\C,R_2^\C)$ such that $\A\rightarrow\C\rightarrow\B$ and $\CSP(\C)$
  is tractable.  We will argue that this would imply finite tractability of 
  $(\inn{1}{3},\NAE)$, which contradicts the result in~\cite{BBKO21}
  (unless P=NP); cf.~Example~\ref{ex:finite}.
  Indeed, the existence of such $\C$ gives the following chain of homomorphisms:
\begin{align}
\label{chain_morphisms_1645_0806}
  \inn{1}{3}=(\{0,1\};R_1^\A)\rightarrow (A;R_1^\A)\rightarrow (C;R_1^\C)\rightarrow(B;R_1^\B)\rightarrow (\{0,1\};R_1^\B)=\NAE
\end{align}
where the first map is the inclusion of $\{0,1\}$ in $A$, the second and the
  third are the maps witnessing $\A\rightarrow\C\rightarrow\B$, and the fourth
  is any map $g:B\rightarrow\{0,1\}$ such that $g(0)=0$ and $g(1)=1$. Let
  $\tilde{\C}=(C;R_1^\C)$. Observe that $\tilde{\C}$ is tractable since the
  inclusion map gives a minion homomorphism $\Pol(\C,\C)\rightarrow
  \Pol(\tilde{\C},\tilde{\C})$, and thus
  $\CSP(\tilde{\C})=\PCSP(\tilde{\C},\tilde{\C})\red\PCSP(\C,\C)=\CSP(\C)$
  by~\cite[Theorem~3.1]{BBKO21}. This proves the claim,
  as~(\ref{chain_morphisms_1645_0806}) established $\inn{1}{3}\rightarrow
  \tilde{\C}\rightarrow\NAE$.
  
Notice that the assignment $f\mapsto g\circ f|_{\{0,1\}^L}$ (where $f$ is a polymorphism of $(\A,\B)$ of arity $L$ and $g$ is the map considered above) yields a minion homomorphism from $\Pol(\A,\B)$ to $\Pol(\inn{1}{3},\NAE)$.
As established in~\cite[Corollary~4.2]{Atserias22:soda}, the template $(\inn{1}{3},\NAE)$ does not have bounded width -- i.e., is not solved by local consistency methods. It follows from~\cite[Lemma~7.5]{BBKO21} that $(\A,\B)$ does not have bounded width either.

\end{example}

The template from Example~\ref{ex:new1} will be proved tractable later (in
Example~\ref{ex:new2}) using our
new algorithm, which we will present next.

\section{The CLAP algorithm}
\label{sec:clap}

Let $(\A,\B)$ be a PCSP template with signature $\sigma$ and let $\X$ be an
instance of $\PCSP(\A,\B)$. Without loss of generality, we assume that $\sigma$
contains a unary symbol $\un$ such that $\un^\X=X$, $\un^\A=A$, and $\un^\B=B$.
If this is not the case, the signature and the instance can be extended without
changing the set of solutions.
Our algorithm -- the combined \textbf{C}B\textbf{L}P+\textbf{A}I\textbf{P} algorithm ($\CLAP$), presented in
Algorithm~\ref{alg:clap} and discussed below -- builds on $\BLP$~\cite{BBKO21} and
$\BLP+\AIP$~\cite{bgwz20}. 

$\CLAP$ works in two stages. In the first stage, it runs $\CBLP$; i.e., a modified version of the singleton arc-consistency algorithm (cf.~\cite{DB97}) where $(i)$ the ``arc-consistency'' part is replaced by $\BLP$, and $(ii)$ the ``singleton'' part is boosted by requiring that every constraint-assignment pair (as opposed to every variable-value pair)
is fixed at each iteration. In the second stage, it refines $\CBLP$ by doing an additional sanity check: At least one of the solutions computed by $\CBLP$ should be compatible with a solution of $\AIP$. As in~\cite{bgwz20}, this second stage requires that the $\AIP$ solution should only use those variables from the $\CBLP$ solution that have nonzero weight. There are two equivalent ways to enforce this requirement: Either by storing the nonzero variables at each iteration of $\CBLP$ in the first stage of the algorithm, or by simply running $\BLP+\AIP$ as a black box in the second stage of the algorithm. We adopt the latter option to achieve a simpler presentation. Concretely, the first stage of $\CLAP$ is performed by initialising the sets $S_{\bx,R}$ of constraint-assignment pairs to the entire relation $R^\A$, and then progressively shrinking these sets by cycling over all constraint-assignment pairs and removing a pair whenever it yields an infeasible $\BLP$. The second stage, that occurs if all sets $S_{\bx,R}$ are nonempty, is performed by cycling over each feasible constraint-assignment pair and running $\BLP+\AIP$ on it. As soon as one constraint-assignment pair is accepted by $\BLP+\AIP$, the algorithm terminates and outputs $\textsc{Yes}$. If no constraint-assignment pair is accepted, the algorithm outputs $\textsc{No}$.

As in Appendix~\ref{app:relaxations}, where $\BLP$, $\AIP$, and $\BLP+\AIP$ are presented in full detail for completeness, by $\lambda_{\bx,R}(\ba)$ we denote the variable of $\BLP(\X,\A)$ associated with $\bx\in R^\X$ and $\ba\in R^\A$, where $R\in\sigma$.
The algorithm has polynomial time complexity in the size of the input instance: Letting $g=\sum_{R\in\sigma}|R^\X||R^\A|$, $\mathcal{O}(g^2)$ $\BLP$ calls and $\mathcal{O}(g)$ $\BLP+\AIP$ calls occur.
We say that $\CLAP$ \emph{accepts} an instance $\X$ of $\PCSP(\A,\B)$ if
Algorithm~\ref{alg:clap} returns \textsc{Yes}.
We say that $\CLAP$ \emph{solves} $\PCSP(\A,\B)$ if, for every instance $\X$ of
$\PCSP(\A,\B)$, we have (i) if $\X\to \A$ then $\CLAP$ accepts $\X$, and (ii)
if $\X$ is accepted by $\CLAP$ then $\X\to\B$.

\begin{algorithm}[tbh]
	\SetAlgoLined
  \KwIn{\quad\ \ \ an instance $\X$ of $\PCSP(\A,\B)$ of signature $\sigma$}
  \KwOut{\quad \textsc{yes} if $\X\to\A$ and \textsc{no} if $\X\not\to\B$}
  \medskip
  \For{$R\in \sigma$, $\bx\in R^\X$}{
  	set $S_{\bx,R}:=R^\A$\;
  }
  \Repeat{no set $S_{\bx,R}$ is changed}{
    \For{$R\in \sigma, \bx\in R^\X$, $\ba\in S_{\bx,R}$}{
      \If{$\BLP(\X,\A)$ with $\lambda_{\bx,R}(\ba)=1$ and $\lambda_{\bx',R'}(\ba')=0$ for every $R'\in\sigma$, $\bx'\in R'^\X$, and $\ba'\not\in S_{\bx',R'}$ is not feasible}{remove $\ba$ from $S_{\bx,R}$\;}
    }
  }
\eIf{some $S_{\bx,R}$ is empty}{return \textsc{No};}{
 \For{$R\in \sigma, \bx\in R^\X$, $\ba\in S_{\bx,R}$}
  {\If{$\BLP+\AIP(\X,\A)$ with $\lambda_{\bx,R}(\ba)=1$ and $\lambda_{\bx',R'}(\ba')=0$ for every $R'\in\sigma$, $\bx'\in R'^\X$, and $\ba'\not\in S_{\bx',R'}$ is feasible}{return \textsc{Yes};}}
return \textsc{No};
}
\caption{The $\CLAP$ algorithm}
	\label{alg:clap}
\end{algorithm}

\paragraph{Characterisation} 
Our first main result -- Theorem~\ref{thm:main1} -- is a minion-theoretic characterisation of the power of the
$\CLAP$ algorithm. In particular, we will introduce in
Definition~\ref{def:minion} a minion $\ourMinion$ such that, for any $\PCSP$
template $(\A,\B)$, the $\CLAP$ algorithm solves $\PCSP(\A,\B)$ if and only if
there is a minion homomorphism from $\ourMinion$ to $\Pol(\A,\B)$. The two
directions will be proved in Theorems~\ref{minion_homo_implies_CBLP_works}
and~\ref{CBLP_works_implies_minion_homo}, respectively, in
Section~\ref{sec_power_CLAP_minions}. 
Combining Theorem~\ref{minion_homo_implies_CBLP_works} with our second main result -- Theorem~\ref{thm:Hsymm}, proved in
Section~\ref{sec_H_symmetric} -- will then yield a sufficient condition for $\CLAP$ to solve a given $\PCSP$ template, in terms of a weak notion of symmetry for the polymorphisms of the template.

The $L$-ary objects of the minion $\ourMinion$ are pairs $(M,\bmu)$, where $M$
is a matrix with $L$ rows and infinitely many columns
encoding the $\BLP$ computations of $\CLAP$ and $\bmu$
is an $L$-ary vector of integers encoding the $\AIP$ computation of $\CLAP$. The
matrices $M$ in $\ourMinion$ have a special structure, which we call
\enquote{skeletal}. 

\begin{defn}\label{def:skeletal}
  Let $M$ be a $p\times \aleph_0$ matrix with $p\in\N$.
  We say that $M$ is \emph{skeletal} if, for each $j\in [p]$, either
$\be_j^TM=\bzero^T_{\aleph_0}$ or $M\be_i=\be_j$ for some $i\in\N$. 
\end{defn}
\noindent In other words, either the $j$-th
row of $M$ is the zero vector or some column of $M$ is the $j$-th
standard unit vector.
Equivalently, $M$ is skeletal if there exist permutation matrices $P\in\R^{p,p}$ and $Q\in\R^{\aleph_0,\aleph_0}$ such that $\displaystyle
PMQ=\begin{bmatrix}
I_k & \tilde{M}\\
O & O
\end{bmatrix}$
for some $k\leq p$ and some $\tilde M\in \R^{k,\aleph_0}$. The
name indicates that the \enquote{body} of a skeletal matrix (the nonzero rows)
is completely supported by a \enquote{skeleton} (the identity block).

We are now ready to define the minion $\ourMinion$. The $L$-ary objects of
$\ourMinion$ are pairs $(M,\bmu)$, where $M$ is a skeletal matrix of size $L\times\aleph_0$ and $\bmu$ is
an affine vector (i.e., an integer vector whose entries sum up to one) of size $L$. We require that every column of $M$ should be
stochastic and $M$ should have only finitely many
different columns; the latter is formalised in $(c_5)$ in
Definition~\ref{def:minion}, which says that starting from some point all the
columns are equal. We also require a particular relationship between $M$
and $\bmu$ formalised in $(c_4)$. 

\begin{defn}\label{def:minion}
For $L\in\N$, let $\ourMinion^{(L)}$ be the set of pairs $(M,\bmu)$ such that $M\in\Q^{L, \aleph_0}$, $\bmu\in \Z^L$, and the following requirements are met:
\begin{align*}
\begin{array}{ll}
  (c_1)\quad M \mbox{ is entrywise nonnegative};\hspace*{2cm}
  &
(c_4)\quad \supp(\bmu)\subseteq\supp(M\be_1);\\
(c_2)\quad  \bone_L^TM=\bone_{\aleph_0}^T;
  &
  (c_5)\quad \exists t\in \N \mbox{ such that } M\be_i=M\be_t\hspace{.4cm}\forall i\geq t;\\
(c_3)\quad \bone_L^T\bmu=1; 
& (c_6)\quad M \mbox{ is skeletal. }
\end{array}
\end{align*}
We define $\ourMinion$ as the disjoint union of $L$-ary parts,
  $\ourMinion\coloneqq\bigcup_{L\geq 1}\ourMinion^{(L)}$.
\end{defn}
\noindent We defined $\ourMinion$ as a set. For $\ourMinion$ to be a minion, we need to define the
minor operation on $\ourMinion$ and verify that it preserves the structure of
$\ourMinion$. This is easy and done in Section~\ref{sec:minion}.

Our first result is the following characterisation of the power of $\CLAP$.

\begin{thm}\label{thm:main1}
Let $(\A,\B)$ be a PCSP template. Then, $\CLAP$ solves $\PCSP(\A,\B)$ if and only if there is a minion homomorphism from $\ourMinion$ to $\Pol(\A,\B)$.
\end{thm}

\paragraph{H-symmetry}
Our second main result is a sufficient condition on a PCSP template $(\A,\B)$
to guarantee that $\CLAP$ solves $\PCSP(\A,\B)$. The condition is through
symmetries satisfied by polymorphisms of the template. In particular,
in Theorem~\ref{thm:Hsymm} we will show that if $\Pol(\A,\B)$ contains
infinitely many operations that are \enquote{$H$-symmetric} for a suitable
matrix $H$, then there is a minion homomorphism from $\ourMinion$ to
$\Pol(\A,\B)$, and thus $\CLAP$ solves $\PCSP(\A,\B)$ by Theorem~\ref{minion_homo_implies_CBLP_works}.

In order to define the notion of $H$-symmetry, we need a few auxiliary
definitions. A vector $\bw=(w_i)\in\R^p$ is
\emph{tieless} if, for any two indices $i\neq i'\in [p]$, $w_i\neq 0$
$\Rightarrow$ $w_i\neq w_{i'}$. A \emph{tie matrix} is a matrix having integer
nonnegative entries, each of whose columns is a tieless vector. Given an
$m\times p$ tie matrix $H$, we say that a vector $\bv\in \R^p$ is
\emph{$H$-tieless} if $H\bv$ is tieless.

Let $A$ be a finite set, let $L\in\N$, and take $\ba=(a_1,\dots,a_L)\in A^L$.
We define the (multiplicity) vector $\ba^\#$ as the integer vector of size $|A|$ whose $a$-th entry is
$|\{i\in [L]:a_i=a\}|$ for each $a\in A$.

\begin{defn}
\label{def_H_tieless_symmetric}
Let $A,B$ be finite sets, and consider a function $f:A^L\rightarrow B$ for some $L\in\N$. Given an $m\times |A|$ tie matrix $H$, we say that $f$ is \emph{$H$-symmetric} if
\begin{align*}
f_{/\pi}(\ba)=f(\ba) && \forall \pi: [L]\rightarrow [L] \mbox{ permutation}, \hspace{.3cm}\forall \ba\in A^L
\mbox{ such that }\ba^\# \mbox{ is $H$-tieless}.
\end{align*}
\end{defn}

\noindent
Our second result is the following sufficient condition for tractability of PCSPs.

\begin{thm}
\label{thm:Hsymm}
Let $(\A,\B)$ be a PCSP template and suppose $\Pol(\A,\B)$ contains
  $H$-symmetric operations of arbitrarily large arity for some $m\times |A|$ tie
  matrix $H$, $m\in\N$. Then there exists a minion homomorphism from
  $\ourMinion$ to $\Pol(\A,\B)$. 
\end{thm}

Recall from Definition~\ref{def:skeletal} the notion of a skeletal matrix. As it will be clear from the rest of the paper, the \enquote{skeleton} represents the 
link between $\CLAP$ and the above-defined notion of $H$-symmetry. Indeed, on the one hand the
presence of the identity block in a skeletal matrix captures the fact that each $\BLP$ solution computed
by $\CLAP$ gives probability $1$ to some constraint-assignment pair and probability $0$ to all other constraint-assignment pairs for the same constraint (cf.~line~6 of~Algorithm~\ref{alg:clap}). On the other hand, Lemma~\ref{lemma_tie_terminator} (stated
and proved in Section~\ref{sec_H_symmetric}) shows that finitely many skeletal matrices can always be
simultaneously reduced to $H$-tieless probability distributions -- which are exactly the distributions on which $H$-symmetric functions are symmetric (cf.~Definition~\ref{def_H_tieless_symmetric}).

We now mention some consequences of Theorem~\ref{thm:Hsymm}. 
First, observe that a vector of size $1$ is always tieless. Hence, if we take any $1\times |A|$ integer nonnegative
matrix as $H$, we have that $H$ is a tie matrix and $\ba^\#$ is $H$-tieless for each tuple $\ba$ in
the domain of $f$; therefore, for such an $H$, $f$ being $H$-symmetric reduces to
$f$ being symmetric. On the other hand, having Definition~\ref{def_H_tieless_symmetric} in mind, adding rows to $H$ increases the chance
that $H\ba^\#$ has some ties, in which case $f$ is released from the
requirement of being symmetric on $\ba$. In this sense, $H$ encodes the
\enquote{exceptions to symmetry} that $f$ is allowed to have: The more rows $H$ has, the stronger Theorem~\ref{thm:Hsymm} becomes. If, for
instance, $H$ is the identity matrix of order $|A|$, then an $H$-symmetric
operation needs to be symmetric only on those tuples where each entry occurs
with a different multiplicity. A very special example of such an
$I_{|A|}$-symmetric operation is a function $f$ that returns (the homomorphic image of) the
\emph{most-frequent} entry in the input tuple whenever it is unique, and, in
any other case, $f$ is, say, (the homomorphic image of) a projection. Other, more creative choices for $H$ allow capturing operations having more complex exceptions to symmetry, as shown in Example~\ref{ex:new2}.
 
Theorems~\ref{thm:main1} and~\ref{thm:Hsymm} together
establish that the $\CLAP$ algorithm solves any PCSP template admitting
arbitrarily large polymorphisms having some exceptions to
symmetry that can be encoded via a tie matrix.  

The importance of the next example lies in the fact that it provably separates $\CLAP$ from finite tractability and $\BLP+\AIP$; i.e., there are PCSP templates solvable by $\CLAP$ that are not finitely tractable and not solvable by $\BLP+\AIP$.

\begin{example}
\label{ex:new2}
Recall the PCSP template $(\A,\B)$ from Example~\ref{ex:new1}, where it was
  shown that $\PCSP(\A,\B)$ is not finitely tractable and not solved by the $\BLP+\AIP$ algorithm from~\cite{bgwz20}.
  We will show that $\PCSP(\A,\B)$ is solved by $\CLAP$.
  
  Take $L\in\N$ and consider the function $f:A^L\rightarrow B$ defined as follows: For $\ba=(a_1,\dots,a_L)\in A^L$,
\begin{itemize}
\item
if $\ba\in \{0,1\}^L$, look at $\ba^\#_1$, i.e., the multiplicity of $1\in A$ in the tuple $\ba$;
\begin{itemize}
\item[$\ast$]
if $\ba^\#_1<\frac{L}{3}$, set $f(\ba)=0$;
\item[$\ast$]
if $\ba^\#_1>\frac{L}{3}$, set $f(\ba)=1$;
\item[$\ast$]
if $\ba^\#_1=\frac{L}{3}$, set $f(\ba)=a_1$;
\end{itemize}
\item
if $\ba\in \{2,3,4,5,6\}^L$,
\begin{itemize}
\item[$\ast$]
if there is a unique element $a\in A$ having maximum multiplicity in $\ba$, set $f(\ba)=a$;
\item[$\ast$]
if there is more than one element of $A$ having maximum multiplicity in $\ba$, set $f(\ba)=a_1$;
\end{itemize}
\item
  otherwise, set $f(\ba)=0$.\footnote{Assigning any value in
    $\{0,\ldots,6\}$ to $f(\ba)$ would work here.}
\end{itemize}
We claim that $f\in\Pol(\A,\B)$. To see that $f$ preserves $R_1$, consider a
  tuple $\brho=(\br_1,\dots,\br_L)$ of elements of $R_1^\A$, where $\br_i=(a_i,b_i,c_i)$ for $i\in [L]$. We shall let $\ba=(a_1,\dots,a_L)$, $\bb=(b_1,\dots,b_L)$, and $\bc=(c_1,\dots,c_L)$. Notice that 
\begin{align}
\label{eqn_sum_al_be_ga_1006_1819}
\ba^\#_1+\bb^\#_1+\bc^\#_1=L. 
\end{align}
If $f(\ba)=f(\bb)=f(\bc)=0$, then $\ba^\#_1\leq\frac{L}{3}$, $\bb^\#_1\leq\frac{L}{3}$, and $\bc^\#_1\leq\frac{L}{3}$; by~\eqref{eqn_sum_al_be_ga_1006_1819}, this implies that $\ba^\#_1=\bb^\#_1=\bc^\#_1=\frac{L}{3}$. Hence, $(0,0,0)=(f(\ba),f(\bb),f(\bc))=(a_1,b_1,c_1)=\br_1\in R_1^\A$, a contradiction. Similarly, $f(\ba)=f(\bb)=f(\bc)=1$ would yield $\ba^\#_1\geq\frac{L}{3}$, $\bb^\#_1\geq\frac{L}{3}$, and $\bc^\#_1\geq\frac{L}{3}$; again by~\eqref{eqn_sum_al_be_ga_1006_1819}, this implies that $\ba^\#_1=\bb^\#_1=\bc^\#_1=\frac{L}{3}$, hence $(1,1,1)=(f(\ba),f(\bb),f(\bc))=(a_1,b_1,c_1)=\br_1\in R_1^\A$, also a contradiction. We conclude that $f(\brho)=(f(\ba),f(\bb),f(\bc))\in R_1^\B$, thus showing that $f$ preserves $R_1$. 

As for $R_2$, let $\brho=(\br_1,\dots,\br_L)$ be a tuple of elements of $R_2^\A$, where $\br_i=(a_i,b_i)$ for $i\in [L]$, and let $\ba=(a_1,\dots,a_L)$ and $\bb=(b_1,\dots,b_L)$. The directed graph having vertex set $\{2,3,4,5,6\}$ and edge set $R_2^\A=R_2^\B$ consists of the disjoint union of a directed $2$-cycle and a directed $3$-cycle and, hence, all of its vertices have in-degree and out-degree one. As a consequence, the multiplicity of a directed edge $(a,b)$ in the tuple $\brho$ equals both the multiplicity of $a$ in $\ba$ and the multiplicity of $b$ in $\bb$. Therefore, if the tuple $\brho$ has a unique element $\br=(a,b)$ with maximum multiplicity, then $f(\brho)=(f(\ba),f(\bb))=(a,b)=\br\in R_2^\B$. Otherwise, $f(\brho)=(a_1,b_1)=\br_1\in R_2^\B$. This shows that $f$ preserves $R_2$, too, and is thus a polymorphism of $(\A,\B)$.

Consider the matrix $H=\diag(1,2,1,1,1,1,1)$, and observe that $H$ is a tie matrix. We claim that $f$ is $H$-symmetric. Let $\pi:[L]\rightarrow [L]$ be a permutation, and take a tuple $\ba=(a_1,\dots,a_L)\in A^L$ such that $\ba^\#$ is $H$-tieless; i.e., the vector $H\ba^\#=(\ba^\#_0,2\ba^\#_1,\ba^\#_2,\ba^\#_3,\ba^\#_4,\ba^\#_5,\ba^\#_6)$ is tieless. Write $\tilde{\ba}=(a_{\pi(1)},\dots,a_{\pi(L)})$, and observe that $\tilde{\ba}^\#=\ba^\#$.
\begin{itemize}
\item
If $\ba\in\{0,1\}^L$, we get $\ba^\#_0\neq2\ba^\#_1$; since $\ba^\#_0+\ba^\#_1=L$, this gives $2\ba^\#_1\neq L-\ba^\#_1$ so that $\ba^\#_1\neq \frac{L}{3}$. As a consequence, $f(\ba)=f(\tilde{\ba})$.
\item
If $\ba\in\{2,3,4,5,6\}^L$, the condition above implies that the tuple $(\ba^\#_2,\ba^\#_3,\ba^\#_4,\ba^\#_5,\ba^\#_6)$ has a unique maximum element and, hence, there is a unique element $a$ of $A$ having maximum multiplicity in $\ba$ (and in $\tilde{\ba}$). Therefore, $f(\ba)=a=f(\tilde{\ba})$.
\item
If $\ba\not\in \{0,1\}^L\cup \{2,3,4,5,6\}^L$, then $f(\ba)=0=f(\tilde{\ba})$.
\end{itemize}
We conclude that, in each case, $f(\ba)=f(\tilde{\ba})=f_{/\pi}(\ba)$, which means that $f$ is $H$-symmetric.
By Theorems~\ref{thm:main1} and~\ref{thm:Hsymm}, $\CLAP$ solves $\PCSP(\A,\B)$.
\end{example}

\begin{rem}
Consider the minion $\Mblpaip$ from~\cite{bgwz20} (cf.~Appendix~\ref{subsec:blp-aip}). A direct consequence of
Example~\ref{ex:new2}, Theorem~\ref{thm:main1}, and~\cite[Lemma 5.4]{bgwz20} is that there is no
minion homomorphism from $\Mblpaip$ to $\ourMinion$. On the other hand, the function
\begin{align*}
\vartheta:\ourMinion&\rightarrow\Mblpaip\\
(M,\bmu)&\mapsto (M\be_1,\bmu)
\end{align*}
is readily seen to be a minion homomorphism. It follows that $\CLAP$ solves any
  PCSP template solved by $\BLP+\AIP$ (as is also clear from the description of
  the two algorithms).
\end{rem}

\begin{rem}
Similar to~\cite{bgwz20}, the assumption in Theorem~\ref{thm:Hsymm} can be weakened as follows: Instead of requiring $H$-symmetric polymorphisms of arbitrarily large arity, it turns out to be enough requiring $H$-block-symmetric polymorphisms of arbitrarily large width, where the definition of an $H$-block-symmetric operation mirrors that of a block-symmetric operation in~\cite{bgwz20}. The proof of this possibly stronger result is very similar to that of Theorem~\ref{thm:Hsymm}. For completeness, we include it in Appendix~\ref{sec:block_sym_pol}.
We point out that we do not know whether the condition in Theorem~\ref{thm:Hsymm} (or the possibly weaker condition based on $H$-block-symmetric polymorphisms) is necessary for tractability via CLAP, but we suspect it is not.
\end{rem}

\begin{rem}
A possibly stronger version of the $\CLAP$ algorithm consists in running $\BLP+\AIP$ (instead of just $\BLP$) at each iteration in the \textbf{for} loop in lines $5$--$9$ of Algorithm~\ref{alg:clap}, and then removing the additional \textbf{for} loop in lines $14$--$18$. This algorithm can be called C($\BLP+\AIP$). An analysis entirely analogous to the one presented in this paper shows that the power of C($\BLP+\AIP$) is captured by the minion $\tilde{\ourMinion}$ defined like $\ourMinion$ with the following difference: The $L$-ary elements of $\tilde{\ourMinion}$ are pairs $(M,N)$, where $M$ is as in $\ourMinion$ while $N$ is an integer matrix of the same size as $M$ taking the role of $\bmu$ (in particular, $N$ satisfies the \enquote{refinement condition} $\supp(N\be_i)\subseteq\supp(M\be_i)$ $\forall i\in\N$, analogous to $(c_4)$ in Definition~\ref{def:minion}). A possible direction for future research is to investigate whether the richer structure of $\tilde{\ourMinion}$ can be exploited to obtain a stronger version of Theorem~\ref{thm:Hsymm}. 
\end{rem}

\begin{rem}
For CSPs, the characterisation of bounded width~\cite{Barto14:jacm,Bulatov09:width} and its collapse~\cite{Barto14:jloc} was preceded by a characterisation of width-1 CSPs~\cite{Feder98:monotone,Dalmau99} and the collapse of width~2 to width~1~\cite{Dalmau09}. Thus the difference between width-1 CSPs and bounded-width CSPs is well understood. BLP and SBLP are the (convex relaxation) analogues of width~1 and SAC, respectively, and SAC solves all bounded-width CSPs~\cite{Kozik21:sicomp}. Therefore, a natural question is whether a similar analysis can cast light on the difference in power between BLP on one side, and SBLP (and thus perhaps also of CBLP and CLAP) on the other side. We remark on two obstacles: Firstly, BLP is strictly more powerful than width~1 for CSPs~\cite{Kun16}. Secondly, a good characterisation of the power of SBLP (and stronger algorithms studied in the present paper) would imply that these algorithms solve, in the special case of CSPs, all bounded-width CSPs -- a non-trivial result implied by~\cite{Kozik21:sicomp}.
\end{rem}

\section{The power of the CLAP algorithm}
\label{sec_power_CLAP_minions}

The goal of this section is to prove Theorem~\ref{thm:main1}. In
Section~\ref{sec:minion}, we will verify that $\ourMinion$, which appears in
the statement of Theorem~\ref{thm:main1}, is indeed a minion. In
Sections~\ref{sec_compactness_argument} and~\ref{sec:cnd}, we will establish a
compactness argument and present a condition that captures
$\CLAP$, respectively; both will be needed in the proof of
Theorem~\ref{thm:main1}. The two directions of Theorem~\ref{thm:main1} will be
then proved in Section~\ref{sec:proof}.

\medskip
Minions are not only useful for capturing the complexity of PCSPs but also for
characterising the power of algorithms. This will be done by using the concept
of the free structure generated, for a given minion, by a relational
structure~\cite{bgwz20} (cf.~\cite[Definition~4.1]{BBKO21} for the definition in
the special case of minions of functions).
\begin{defn}\label{def:free}
Let $\mathscr{M}$ be a minion and let $\A$ be a (finite) relational structure with signature $\sigma$. The \emph{free structure} $\mathbb{F}_{\mathscr{M}}(\A)$ is a relational structure with domain $\mathscr{M}^{(|A|)}$ (potentially infinite) and signature $\sigma$. Given a relation $R\in\sigma$ of arity $k$, a tuple $(M_1,\dots,M_k)$ of elements of $\mathscr{M}^{(|A|)}$ belongs to $R^{\mathbb{F}_{\mathscr{M}}(\A)}$ if and only if there is some $Q\in \mathscr{M}^{(|R^\A|)}$ such that $M_i=Q_{/\pi_i}$ for each $i\in[k]$, where $\pi_i:R^\A\to A$ maps $\textbf{a}\in R^\A$ to its $i$-th coordinate $a_i$.
\end{defn}
\noindent The next result will be useful to establish the connection between our
algorithm $\CLAP$, presented in Section~\ref{sec:clap}, and the minion $\ourMinion$.

\begin{lem}\label{lem:free}
Let $\mathscr{M}$ be a minion and let $(\A,\B)$ be a PCSP template. Then there is a minion homomorphism from $\mathscr{M}$ to $\Pol(\A,\B)$ if and only if $\mathbb{F}_\mathscr{M}(\A)\to\B$.
\end{lem}
\noindent The proof of Lemma~\ref{lem:free} is based on that of~\cite[Lemma~4.4]{BBKO21},
which proves one-to-one correspondence but only for minions of functions. For
completeness, we prove Lemma~\ref{lem:free} in Appendix~\ref{sec:free-proof}.

\subsection{$\ourMinion$ is a minion}\label{sec:minion}

The minor operation on $\ourMinion$ is naturally defined via a matrix
multiplication with a matrix that encodes the minor map. For a function
$\pi:[L]\rightarrow [L']$, let ${P_\pi}$ be the $L'\times L$ matrix whose
$(i,j)$-th entry is $1$ if $\pi(j)=i$, and $0$ otherwise. Note that
${P_\pi^T}\bone_{L'}=\bone_{L}$ and, for each $i\in [L']$, ${P_\pi^T}
\be_i=\sum_{j\in\pi^{-1}(i)}\be_j$. 

\begin{defn}
For $(M,\bmu)\in \ourMinion^{(L)}$, we define
$M_{/\pi}={P_\pi} M$ and $\bmu_{/\pi}= {P_\pi}\bmu$,
and we let the minor of $(M,\bmu)$ with respect to $\pi$ be 
$(M,\bmu)_{/\pi}\coloneqq (M_{/\pi},\bmu_{/\pi})$.
\end{defn}
\noindent We remark that this definition is consistent
with the minions $\Qconv$ and $\Zaff$ studied in~\cite{BBKO21}, and the minion $\Mblpaip$ studied in~\cite{bgwz20},
cf.~Appendices~\ref{subsec:blp},~\ref{subsec:aip}, and~\ref{subsec:blp-aip}.

\begin{prop}
$\ourMinion$ is a minion.
\end{prop}
\begin{proof}
Write $M=[m_{ij}]$ and $\bmu=(\mu_i)$. Observe that $M_{/\pi}\in \Q^{L',\aleph_0}$ and $\bmu_{/\pi}\in \Z^{L'}$. The requirements $(c_1),(c_2),(c_3)$, and $(c_5)$ are trivially satisfied by $(M,\bmu)_{/\pi}$. As for $(c_4)$, suppose that $\be_i^T{P_\pi}M\be_1=0$ but $\be_i^T{P_\pi}\bmu\neq 0$. It follows that $\mu_j\neq 0$ for some $j\in\pi^{-1}(i)$. Hence, $m_{j1}>0$ and, then, 
\begin{align*}
\be_i^T{P_\pi}M\be_1=\sum_{j'\in \pi^{-1}(i)}\be_{j'}^TM\be_1\geq \be_j^TM\be_1>0,
\\
\end{align*}
which is a contradiction. We now show that $M_{/\pi}$ is skeletal. Choose $j\in [L']$, and suppose that $\be_j^TM_{/\pi}\neq \bzero_{\aleph_0}^T$. We obtain
\begin{align*}
\bzero_{\aleph_0}\neq M^T{P_\pi^T} \be_j
=
\sum_{{\ell}\in \pi^{-1}(j)}M^T\be_{\ell}
\end{align*}
and, in particular, $\exists {\ell}\in \pi^{-1}(j)$ such that $\be_{\ell}^TM\neq \bzero_{\aleph_0}^T$. Since $M$ is skeletal, this implies that $M\be_i=\be_{\ell}$ for some $i\in\N$. This yields
\begin{align*}
M_{/\pi}\be_i={P_\pi}M\be_i={P_\pi}\be_{\ell}=\be_{\pi({\ell})}=\be_j
\end{align*}
as required. Hence, $(c_6)$ is satisfied, too, and $(M,\bmu)_{/\pi}\in \ourMinion^{(L')}$. 

Finally, considering $\tilde\pi:[L']\to[L'']$ and the identity map $\operatorname{id}:[L]\to[L]$,
one readily checks that ${P_{\tilde{\pi}\circ\pi}}={P_{\tilde\pi}} {P_{\pi}}$ and ${P_{\operatorname{id}}}=I_L$. Hence, the minor operations defined above satisfy the requirements of Definition~\ref{defn_minion}.
\end{proof}

\subsection{A compactness argument for $\ourMinion$}
\label{sec_compactness_argument}
The set $\ourMinion^{(L)}$ of the $L$-ary objects in $\ourMinion$ is infinite
unless $L=1$. As a consequence, given a relational structure $\A$ whose domain
has size at least $2$, the free structure $\mathbb{F}_\ourMinion(\A)$ has an
infinite domain. We now describe a standard compactness argument analogous
to~\cite[Remark 7.13]{BBKO21} that will circumvent this inconvenience.

For $D,L\in\N$, consider the set 
\begin{align*}
\ourMinion_D^{(L)}=\{(M,\bmu)\in\ourMinion^{(L)}: DM \mbox{ is entrywise integer, } M\be_i=M\be_D \hspace{.2cm}\forall i\geq D, \mbox{ and }\bone_L^T|\bmu|\leq D\},
\end{align*}
where $|\bmu|$ denotes the vector whose entries are the absolute values of the entries of $\bmu$.
Since $\ourMinion_D^{(L)}$ is unambiguously determined by $L\times (D+1)$ integer numbers belonging to the set $\{-D,\dots,D\}$, it is finite. Observe that the set $\ourMinion_D=\bigcup_{L\in\N}\ourMinion_D^{(L)}$ is closed under taking minors. Indeed, given $(M,\bmu)\in \ourMinion_D^{(L)}$ and $\pi:[L]\rightarrow [L']$, $D{P_\pi}M={P_\pi}DM$ is entrywise integer, ${P_\pi}M\be_i={P_\pi}M\be_D$ $\forall i\geq D$, and $\bone_{L'}^T|{P_\pi}\bmu|\leq \bone_{L'}^T{P_\pi}|\bmu|=\bone_L^T|\bmu|\leq D$. Hence, $\ourMinion_D$ is a sub-minion of $\ourMinion$. Observe also that $\ourMinion=\bigcup_{D\in \N}\ourMinion_D$. To see this, take $(M,\bmu)\in \ourMinion^{(L)}$ and suppose that $M\be_i=M\be_t$ $\forall i\geq t$. Let $\tilde{D}$ be a common denominator of the finite set of rational numbers $\{m_{ij}:i\in [L],j\in [t]\}$, so that $\tilde{D}M$ is entrywise integer. Let also $\hat D=\bone_L^T|\bmu|$. Then, $(M,\bmu)\in \ourMinion_{t\tilde{D}\hat{D}}$. 
\begin{prop}
\label{prop_compactness_argument}
Let $\mathscr{M}$ be a minion such that $\mathscr{M}^{(L)}$ is finite for each $L\in\N$, and suppose that there exist minion homomorphisms $\xi_D:\ourMinion_D\rightarrow\mathscr{M}$ for each $D\in \N$. Then there exists a minion homomorphism $\zeta:\ourMinion\rightarrow\mathscr{M}$. 
\end{prop}
\begin{proof}
For $D\in\N$, let $\ourMinion_{\{D\}}=\bigcup_{L\leq D}\ourMinion_{D!}^{(L)}$. Observe that $\ourMinion_{\{D\}}$ is a finite set and  $\ourMinion_{\{D\}}\subseteq \ourMinion_{\{D+1\}}$. Moreover, $\bigcup_{D\in\N}\ourMinion_{\{D\}}=\bigcup_{D\in\N}\ourMinion_D=\ourMinion$. Indeed, given $D'\in \N$, we have that $\ourMinion_{\{D'\}}\subseteq\ourMinion_{D'!}\subseteq \bigcup_{D\in\N}\ourMinion_D$, and, given $L\in\N$, $\ourMinion_{D'}^{(L)}\subseteq \ourMinion_{(D'L)!}^{(L)}\subseteq \ourMinion_{\{D'L\}}\subseteq \bigcup_{D\in\N}\ourMinion_{\{D\}}$. Consider an infinite rooted tree whose vertices are all the restrictions of the homomorphisms $\xi_D$ to some $\ourMinion_{\{D'\}}$, whose root is the empty mapping, and the parent of a vertex corresponding to a function $\ourMinion_{\{D'+1\}}\rightarrow\mathscr{M}$ is the vertex corresponding to the restriction of the function to $\ourMinion_{\{D'\}}$. This is an infinite connected tree. Moreover, since $\mathscr{M}^{(L)}$ is finite for each $L\in\N$ and since minion homomorphisms preserve the arities, there exist only finitely many distinct restrictions of minion homomorphisms to $\ourMinion_{\{D\}}$; hence, the tree is locally finite. By K\H{o}nig's Lemma, it contains an infinite path, which corresponds to an infinite chain of maps $\zeta_i:\ourMinion_{\{i\}}\rightarrow\mathscr{M}$ such that $\zeta_{i+1}$ extends $\zeta_i$ $\forall i\in\N$. Their union $\zeta:\ourMinion\rightarrow\mathscr{M}$ is then a minion homomorphism.
\end{proof}

\subsection{The CLAP condition}\label{sec:cnd}

Given a finite set $C$, 
consider the set $\mathbb{S}(C)$ of the
rational stochastic vectors of size $|C|$. 
Let ${U}\subseteq C^k$. For $i\in [k]$,
consider the $|C|\times |U|$ matrix ${E}^{({U},i)}$ such that, for $c\in C$ and
$\bc=(c_1,\dots,c_k)\in{U}$, the $(c,\bc)$-th entry of ${E}^{({U},i)}$ is $1$ if
$c_i=c$, and $0$ otherwise. Given $\bxi\in \mathbb{S}({U})$ and $i\in [k]$, we define the $i$-th \emph{marginal} of $\bxi$ as
\begin{align*}
\bxi^{(i)}={{E}^{({U},i)}}\bxi.
\end{align*}
Observe that
\begin{align*}
{\bxi^{(i)}}^T\bone_{|C|}=\bxi^T {{E}^{({U},i)}}^T\bone_{|C|}=\bxi^T\bone_{|{U}|}=1,
\end{align*}
so that $\bxi^{(i)}\in\mathbb{S}(C)$. 
We also define the set $\mathbb{Z}(C)$ of the integer vectors of size $|C|$ whose entries sum up to $1$. Given ${U}\subseteq C^k$, $\bzeta\in \mathbb{Z}({U})$, and $i\in [k]$, we define
\begin{align*}
\bzeta^{(i)}={{E}^{({U},i)}}\bzeta.
\end{align*}
As before, observe that
\begin{align*}
{\bzeta^{(i)}}^T\bone_{|C|}=\bzeta^T {{E}^{({U},i)}}^T\bone_{|C|}=\bzeta^T\bone_{|{U}|}=1,
\end{align*}
so $\bzeta^{(i)}\in\mathbb{Z}(C)$.

Let $\A$ be a relational structure having domain $A$ and signature $\sigma$. We define the relational structures $\mathbb{S}(\A)$ and $\mathbb{Z}(\A)$ as follows:
\begin{itemize}
\item
$\mathbb{S}(\A)$ has domain $\mathbb{S}(A)$ and, for every symbol $R\in\sigma$ of arity $k$,\\ $R^{\mathbb{S}(\A)}=\{(\bxi^{(1)},\dots,\bxi^{(k)}):\bxi\in\mathbb{S}(R^{\A})\}$;
\item
$\mathbb{Z}(\A)$ has domain $\mathbb{Z}(A)$ and, for every symbol $R\in\sigma$ of arity $k$,\\ $R^{\mathbb{Z}(\A)}=\{(\bzeta^{(1)},\dots,\bzeta^{(k)}):\bzeta\in\mathbb{Z}(R^{\A})\}$.
\end{itemize}

\begin{rem}
$\mathbb{S}(\A)$ and $\mathbb{Z}(\A)$ are denoted by $\operatorname{LP}(\A)$ and
  $\operatorname{IP}(\A)$ in~\cite{BBKO21}, respectively. As noted
  in~\cite[Remarks~7.11 and~7.21]{BBKO21}, $\mathbb{S}(\A)$ coincides with the free structure of the minion $\Qconv$ generated by $\A$ and, similarly, $\mathbb{Z}(\A)$ is the free structure of the minion $\Zaff$ generated by $\A$. (See Appendices~\ref{subsec:blp} and~\ref{subsec:aip} for the definitions of $\Qconv$ and $\Zaff$, respectively.) In particular, given a relational
  structure $\X$ with signature $\sigma$, $\BLP$ accepts $\X$ as an instance of
  $\CSP(\A)$ if and only if $\X\rightarrow \mathbb{S}(\A)$; similarly, $\AIP$
  accepts $\X$ as an instance of $\CSP(\A)$ if and only if $\X\rightarrow \mathbb{Z}(\A)$.
\end{rem}

\begin{rem}
\label{rem_canonical_homs_A_SA_ZA}
The assignment $f:a\mapsto \be_a$ for each $a\in A$ yields both a canonical
  homomorphism from $\A$ to $\mathbb{S}
(\A)$ and a canonical homomorphism from $\A$ to $\mathbb{Z}
(\A)$. Indeed, for $R\in\sigma$ of arity $k$ and $\ba=(a_1,\dots,a_k)\in R^\A$,
\begin{align*}
f(\ba)=(\be_{a_1},\dots,\be_{a_k})=({{E}^{(R^\A,1)}}\be_\ba,\dots,{{E}^{(R^\A,k)}}\be_\ba)
\end{align*}
which belongs to both $R^{\mathbb{S}(\A)}$ and $R^{\mathbb{Z}(\A)}$ since $\be_\ba\in\mathbb{S}(R^\A)\cap\mathbb{Z}(R^\A)$.
\end{rem}

 In Proposition~\ref{prop_2208_2504}, we characterise the instances of a given PCSP template for which the $\CLAP$ algorithm returns \textsc{yes} in terms of the condition described in the following definition.
\begin{defn}
\label{defn_CBLPC_2504_2212}
Let $(\A,\B)$ be a PCSP template, where $\A$ and $\B$ have signature $\sigma$.
  Given an instance $\X$ of $\PCSP(\A,\B)$, we say that $\X$ has the $\CLAP$
  \emph{condition} if the following holds: $\forall R\in\sigma$ of arity $k$
  $\exists s^R:R^\X\rightarrow\power(R^\A)\setminus\{\emptyset\}$ such that

\begin{itemize}
\item[$(I)$]
$\forall \bx=(x_1,\dots,x_k)\in R^\X, \forall \ba=(a_1,\dots,a_k)\in s^R(\bx)$ there is a homomorphism $h_{\bx,\ba}:\X\rightarrow\mathbb{S}(\A)$ that satisfies:
\begin{enumerate}
\item
$h_{\bx,\ba}(x_i)=\be_{a_i}$ $\forall i\in [k]$;
\item
$\forall \tilde R\in \sigma$ of arity $\tilde k$, $\forall \tilde\bx=(\tilde{x}_1,\dots,\tilde{x}_{\tilde{k}})\in \tilde R^\X$  $\exists\bxi\in\mathbb{S}(\tilde R^\A)$ such that
\begin{itemize}
\item[$\ast$]
$h_{\bx,\ba}(\tilde x_i)={{E}^{(\tilde R^\A,i)}}\bxi\hspace{2cm}\forall i\in [\tilde k]$;
\item[$\ast$]
$\supp(\bxi)\subseteq s^{\tilde R}(\tilde \bx)$.
\end{itemize}
\end{enumerate}
\item[$(II)$]
$\exists \bar{R}\in\sigma, \bar{\bx}\in \bar{R}^\X, \bar{\ba}\in s^{\bar{R}}(\bar{\bx})$ such that there is a homomorphism $g:\X\rightarrow\mathbb{Z}(\A)$ that satisfies:
\begin{itemize}
\item[$1'.$]
$\forall \tilde R\in \sigma$ of arity $\tilde k$, $\forall \tilde\bx=(\tilde{x}_1,\dots,\tilde{x}_{\tilde{k}})\in \tilde R^\X$  $\exists\bxi\in\mathbb{S}(\tilde R^\A)$, $\exists\bzeta\in\mathbb{Z}(\tilde R^\A)$ such that
\begin{itemize}
\item[$\ast$]
$h_{\bar{\bx},\bar{\ba}}(\tilde x_i)={{E}^{(\tilde R^\A,i)}}\bxi\hspace{2cm}\forall i\in [\tilde k]$;
\item[$\ast$]
$g(\tilde x_i)={{E}^{(\tilde R^\A,i)}}\bzeta\hspace{2.44cm}\forall i\in [\tilde k]$;
\item[$\ast$]
$\supp(\bzeta)\subseteq \supp(\bxi)\subseteq s^{\tilde{R}}(\tilde{\bx})$.
\end{itemize}
\end{itemize}
\end{itemize}
\end{defn} 

\begin{prop}
\label{prop_2208_2504}
Given an instance $\X$ of $\PCSP(\A,\B)$, $\CLAP$ accepts $\X$ if and only if $\X$ has the $\CLAP$ condition.
\end{prop}
\begin{proof}
Suppose that $\CLAP$ accepts $\X$ and let $\{S_{\bx,R}: R\in\sigma,\bx\in R^\X\}$
  be the family of sets generated by the algorithm at termination. For each
  $R\in \sigma$, consider the map
  $s^R:R^\X\rightarrow\power(R^\A)\setminus\{\emptyset\}$ defined by
  $s^R(\bx)=S_{\bx,R}$. For each $\bx \in R^\X, \ba\in s^R(\bx)$, consider the
  corresponding solution to $\BLP(\X,\A)$ generated by the algorithm. Letting $w_x$
  be the probability distribution on $A$ associated with $x\in X$ in the linear
  program, we observe that the assignment $x\mapsto w_x$ yields a homomorphism
  (call it $h_{\bx,\ba}$) from $\X$ to $\mathbb{S}(\A)$ that satisfies the
  requirement $1$. Moreover, letting $\bxi$ be the probability distribution
  associated with a constraint $\tilde\bx\in \tilde R^\X$ for some $\tilde
  R\in\sigma$, observe that $h_{\bx,\ba}$ also satisfies the requirement $2$.
  Finally, let $\bar{R}\in\sigma, \bar{\bx}\in \bar{R}^\X, \bar{\ba}\in
  S_{\bar{\bx},\bar{R}}$ be such that the condition in the \textbf{if} statement of
  line 15 of Algorithm~\ref{alg:clap} is met. Then $1'$ follows from the description of $\BLP+\AIP$. 

The converse implication follows almost analogously, except for the following
  subtlety. The $\BLP+\AIP$ algorithm requires that the $\BLP$ solution should
  be picked from the \emph{relative interior} of the polytope of the feasible
  solutions (cf.~Algorithm~\ref{alg:blp-aip} in Appendix~\ref{subsec:blp-aip}). However, the homomorphism $h_{\bar{\bx},\bar{\ba}}$ whose existence witnesses part $(II)$ of the $\CLAP$ condition may correspond to a $\BLP$ solution that is not in the relative interior of the polytope $P$ of the feasible solutions of $\BLP(\X,\A)$ satisfying $\lambda_{\bar\bx,\bar R}(\bar\ba)=1$ and $\lambda_{\bx',R'}(\ba')=0$ for every $R'\in\sigma$, $\bx'\in R'^\X$, and $\ba'\not\in S_{\bx',R'}$. If that is the case, the algorithm would not consider $(h_{\bar{\bx},\bar{\ba}},g)$ as a solution for $\BLP+\AIP$. However, letting $h'$ be a solution in the relative interior of $P$, the conditions $(I)$ and $(II)$ of $\CLAP$ are still satisfied if we let $h'$ replace $h_{\bar{\bx},\bar{\ba}}$; and, in this case, the homomorphisms witnessing the $\CLAP$ condition do correspond to solutions found by the $\CLAP$ algorithm.\footnote{Another way to phrase this is by saying that the existence of a pair $(h,g)$ of homomorphisms such that each variable for $g$ is zero whenever the corresponding variable for $h$ is zero is equivalent to the existence, for \emph{any} $h'$ in the nonempty relative interior of the polytope of solutions of the $\BLP$, of a solution $g'$ of $\AIP$ that sets to zero any variable that is zero in $h'$. This is implicit in the analysis in~\cite{bgwz20}.}
Hence, $\CLAP$ accepts $\X$.
\end{proof}

\subsection{Proof of Theorem~\ref{thm:main1}}\label{sec:proof}

Our first goal is to prove the following.

\begin{thm}
\label{minion_homo_implies_CBLP_works}
If there is a minion homomorphism from $\ourMinion$ to $\Pol(\A,\B)$ then $\CLAP$ solves $\PCSP(\A,\B)$.
\end{thm}
\begin{proof}

Let $\X$ be an instance of $\PCSP(\A,\B)$. 

First we show that if $\X\to\A$ then $\CLAP$ accepts $\X$, which is the easy direction.
Consider a homomorphism $f:\X\rightarrow\A$. Given $R\in\sigma$ of arity $k$ and
  $\bx=(x_1,\dots,x_k)\in R^\X$, let $s^R(\bx)=\{f(\bx)\}$. For $\bx\in R^\X$
  and $\ba=(a_1,\dots,a_k)=f(\bx)\in s^R(\bx)$, let $h_{\bx,\ba}:\X\rightarrow\mathbb{S}(\A)$ be the homomorphism obtained by composing $f$ with the canonical homomorphism from $\A$ to $\mathbb{S}(\A)$ of Remark~\ref{rem_canonical_homs_A_SA_ZA} -- i.e., $h_{\bx,\ba}(x)=\be_{f(x)}$ $\forall x\in X$. Observe that 
$
h_{\bx,\ba}(x_i)=\be_{f(x_i)}=\be_{a_i}  
$ for any $i\in [k]$ and, given $\tilde{R}\in\sigma$ of arity $\tilde{k}$ and $\tilde{\bx}=(\tilde{x}_1,\dots,\tilde{x}_{\tilde{k}})\in \tilde{R}^\X$, setting $\bxi=\be_{f(\tilde{\bx})}$ yields $h_{\bx,\ba}(\tilde{x}_i)=\be_{f(\tilde{x}_i)}={{E}^{(\tilde{R}^\A,i)}}\be_{f(\tilde{\bx})}={{E}^{(\tilde{R}^\A,i)}}\bxi$ for any $i\in [\tilde{k}]$, and $\supp(\bxi)=\supp(\be_{f(\tilde{\bx})})=\{f(\tilde{\bx})\}=s^{\tilde{R}}(\tilde{\bx})$. This shows that part $(I)$ of Definition~\ref{defn_CBLPC_2504_2212} is satisfied. As for part $(II)$, choose any $\bar{R}\in\sigma$ and $\bar{\bx}\in\bar{R}^\X$, let
  $\bar{\ba}=f(\bar{\bx})$, and consider the homomorphism
  $g:\X\rightarrow\mathbb{Z}(\A)$ obtained by composing $f$ with the canonical homomorphism from $\A$ to $\mathbb{Z}(\A)$ of Remark~\ref{rem_canonical_homs_A_SA_ZA} -- i.e., $g(x)=\be_{f(x)}$ $\forall x\in X$. Given $\tilde{R}\in\sigma$ of arity $\tilde{k}$ and $\tilde{\bx}=(\tilde{x}_1,\dots,\tilde{x}_{\tilde{k}})\in \tilde{R}^\X$, setting $\bxi=\bzeta=\be_{f(\tilde{\bx})}$ yields $g(\tilde{x_i})=h_{\bar\bx,\bar \ba}(\tilde{x}_i)=\be_{f(\tilde{x}_i)}={{E}^{(\tilde{R}^\A,i)}}\be_{f(\tilde{\bx})}={{E}^{(\tilde{R}^\A,i)}}\bxi={{E}^{(\tilde{R}^\A,i)}}\bzeta$ for any $i\in [\tilde{k}]$, and $\supp(\bzeta)=\supp(\bxi)=\{f(\tilde{\bx})\}=s^{\tilde{R}}(\tilde{\bx})$. It follows that $\X$ has the $\CLAP$ condition. By~Proposition~\ref{prop_2208_2504}, $\CLAP$ accepts $\X$.

Second we show that if $\X$ is accepted by $\CLAP$ then $\X\to\B$.
So, suppose that $\X$ is accepted by $\CLAP$. By Proposition~\ref{prop_2208_2504}, $\X$ has the $\CLAP$ condition. Using the terminology of Definition~\ref{defn_CBLPC_2504_2212}, consider the set $\{h_1,\dots,h_t\}=\{h_{\bx,\ba}:R\in\sigma,\bx\in R^\X,\ba\in s^R(\bx)\}$, where each $h_{\bx,\ba}$ is a homomorphism from $\X$ to $\mathbb{S}(\A)$ described in part $(I)$ of the definition. We also consider the homomorphism $g:\X\rightarrow\mathbb{Z}(\A)$ of part $(II)$ of the definition, corresponding to $\bar{R}\in\sigma, \bar{\bx}\in \bar{R}^\X, \bar{\ba}\in s^{\bar{R}}(\bar{\bx})$. Without loss of generality, we set $h_1=h_{\bar{\bx},\bar{\ba}}$. 

Let $n=|A|$. Given $x\in X$, consider the matrix $M_x\in\Q^{n,\aleph_0}$ and the vector $\bmu_x\in \Z^n$ defined by
\begin{align*}
\begin{array}{ll}
M_x\be_i=h_i(x) & \forall i\in [t],\\
M_x\be_i=h_t(x) & \forall i\in \N\setminus [t],\\
\bmu_x=g(x). &
\end{array}
\end{align*}
We claim that $(M_x,\bmu_x)\in \ourMinion^{(n)}$. The requirements
  $(c_1),(c_2),(c_3)$, and $(c_5)$ in Definition~\ref{def:minion} are easily seen to be satisfied.
  To check that $M_x$ is skeletal, take $a\in A$ and suppose that $\be_a^TM_x\neq
  \bzero_{\aleph_0}^T$. This means that $\be_a^TM_x\be_d\neq 0$ for some $d\in [t]$.
  Hence, $a\in \supp(M_x\be_d)=\supp(h_d(x))$. Recall that we are assuming (with
  no loss of generality) that the signature $\sigma$ of $\X$, $\A$, and $\B$
  contains a unary symbol $\un$ such that
  $\un^\X=X$, $\un^\A=A$ , and
  $\un^\B=B$. Notice that
  ${E}^{(\un^\A,1)}=I_n$. From part $(I)$ of Definition~\ref{defn_CBLPC_2504_2212}, we deduce that $\supp(h_d(x))\subseteq
  s^{\un}(x)$ and, hence, $a\in s^{\un}(x)$.
  We can then take the homomorphism $h_i=h_{x,a}$, which satisfies $h_i(x)=\be_a$,
  that is $M_x\be_i=\be_a$. So, $M_x$ is skeletal and $(c_6)$ is satisfied.
  Finally, to check $(c_4)$, choose $a\in A$ and suppose that $\be_a^TM_x\be_1=0$. Since $M_x\be_1=h_1(x)=h_{\bar{\bx},\bar{\ba}}(x)$, this implies that $a\not\in \supp(h_{\bar{\bx},\bar{\ba}}(x))$. Choosing $\un$ as $\tilde{R}$ and $x$ as $\tilde{\bx}$ in $1'$ of Definition~\ref{defn_CBLPC_2504_2212}, and using again the fact that ${E}^{(\un^\A,1)}=I_n$, we see that $\supp(g(x))\subseteq \supp(h_{\bar{\bx},\bar{\ba}}(x))$. Therefore, $a\not\in \supp(g(x))=\supp(\bmu_x)$. Hence, $(c_4)$ is satisfied, too, and the claim is proved.

Consider the map 
$\gamma:X\rightarrow \ourMinion^{(n)}$ defined by $x\mapsto (M_x,\bmu_x)$. We
  claim that $\gamma$ is a homomorphism from $\X$ to
  $\mathbb{F}_{\ourMinion}(\A)$. With this claim, we can finish the proof. By assumption, there is a minion homomorphism
  from $\ourMinion$ to $\Pol(\A,\B)$. By Lemma~\ref{lem:free} applied to
  $\ourMinion$, we have $\mathbb{F}_{\ourMinion}(\A)\to\B$. Composing $\gamma$
  with this homomorphism yields $\X\rightarrow\B$, as required. It remains to
  establish the claim.

\medskip
\noindent\textbf{Claim:} $\gamma$ is a homomorphism from $\X$ to $\mathbb{F}_{\ourMinion}(\A)$. 
\medskip

\noindent
  Take $R\in\sigma$ of arity $k$, and let $\bx=(x_1,\dots,x_k)\in R^\X$. We need to show that $((M_{x_1},\bmu_{x_1}),\dots,\allowbreak(M_{x_k},\bmu_{x_k}))\in R^{\mathbb{F}_{\ourMinion}(\A)}$. For each $i\in [t]\setminus\{1\}$, consider a probability distribution $\bxi_i\in \mathbb{S}(R^\A)$ corresponding to the homomorphism $h_i$ and witnessing part $2$ in Definition~\ref{defn_CBLPC_2504_2212}. Also, consider the probability distribution $\bxi_1\in \mathbb{S}(R^\A)$ and the integer distribution $\bzeta\in\mathbb{Z}(R^\A)$ corresponding to $h_1$ and $g$, respectively, and witnessing $1'$. We introduce the matrix $Q\in\Q^{|R^\A|,\aleph_0}$ and the vector $\bdelta\in\Z^{|R^\A|}$ defined by 
\begin{align*}
\begin{array}{ll}
Q\be_i=\bxi_i & \forall i\in [t],\\
Q\be_i=\bxi_t & \forall i\in \N\setminus [t],\\
\bdelta=\bzeta.
\end{array}
\end{align*}
We claim that $(Q,\bdelta)\in \ourMinion^{(|R^\A|)}$. The requirements $(c_1), (c_2), (c_3)$, and $(c_5)$ in Definition~\ref{def:minion} are easily seen to be satisfied. Suppose $\be_\ba^TQ\neq \bzero_{\aleph_0}^T$ for some $\ba=(a_1,\dots,a_k)\in R^\A$, so that there exists $d\in [t]$ such that $\be_\ba^TQ\be_d\neq 0$. Hence, $\ba\in \supp(Q\be_d)=\supp(\bxi_d)\subseteq s^R(\bx)$. Pick $h_j=h_{\bx,\ba}$. We have that
\begin{align*}
{{E}^{(R^\A,p)}}\bxi_j=h_j(x_p)=h_{\bx,\ba}(x_p)=\be_{a_p}& &\forall p\in [k].
\end{align*}
Suppose that $\bxi_j\neq \be_\ba$. Then, $\exists \ba'=(a'_1,\dots,a'_k)\in R^\A$ such that $\ba'\neq \ba$ and $\be_{\ba'}^T\bxi_j>0$. Choose $q\in [k]$ such that $a'_q\neq a_q$, and observe that
\begin{align*}
0=\be_{a'_q}^T\be_{a_q}=\be_{a'_q}^T {{E}^{(R^\A,q)}}\bxi_j
%=\left({{E}^{(R^\A,q)}}^T\be_{a'_q}\right)^T\bxi_j
\geq \be_{\ba'}^T\bxi_j>0,
\end{align*}
which is a contradiction. Hence, $Q\be_j=\bxi_j=\be_\ba$. We conclude that $Q$ is skeletal and, therefore, $(c_6)$ is satisfied. Finally, suppose that $\ba\not\in\supp(Q\be_1)=\supp(\bxi_1)$ for some $\ba\in R^\A$. Recalling that $\bxi_1\in\mathbb{S}(R^\A)$ corresponds to the homomorphism $h_1=h_{\bar{\bx},\bar{\ba}}$, it follows from $1'$ that $\supp(\bzeta)\subseteq\supp(\bxi_1)$. Hence, $\ba\not\in\supp(\bzeta)=\supp(\bdelta)$, so that $(c_4)$ is satisfied, too. As a consequence, $(Q,\bdelta)\in \ourMinion^{(|R^\A|)}$, as claimed. 

Now, we need to show that $(M_{x_\alpha},\bmu_{x_\alpha})=(Q,\bdelta)_{/\pi_\alpha}$ for each $\alpha\in [k]$, where $\pi_\alpha:R^\A\rightarrow A$ maps $\ba\in R^\A$ to its $\alpha$-th coordinate. Observe first that, by definition, ${P_{\pi_\alpha}}={E}^{(R^\A,\alpha)}$ for each $\alpha\in [k]$. We see that
\begin{align*}
&Q_{/\pi_\alpha}\be_i={P_{\pi_\alpha}}Q\be_i={{E}^{(R^\A,\alpha)}}Q\be_i={{E}^{(R^\A,\alpha)}}\bxi_i=h_i(x_\alpha)=M_{x_\alpha}\be_i &&\mbox{for }i\in [t]&& \mbox{ and}\\
&Q_{/\pi_\alpha}\be_i={P_{\pi_\alpha}}Q\be_i={P_{\pi_\alpha}}Q\be_t=M_{x_\alpha}\be_t=M_{x_\alpha}\be_i &&\mbox{for }i\in\N\setminus [t],&&\\
\end{align*}
which yields $Q_{/\pi_\alpha}=M_{x_\alpha}$. Moreover,
\begin{align*}
\bdelta_{/\pi_\alpha}={P_{\pi_\alpha}}\bdelta={{E}^{(R^\A,\alpha)}}\bdelta={{E}^{(R^\A,\alpha)}}\bzeta=g(x_\alpha)=\bmu_{x_\alpha}.
\end{align*}
It follows that $(M_{x_\alpha},\bmu_{x_\alpha})=(Q_{/\pi_\alpha},\bdelta_{/\pi_\alpha})=(Q,\bdelta)_{/\pi_\alpha}$.
By Definition~\ref{def:free}, $((M_{x_1},\bmu_{x_1}),\dots, \linebreak (M_{x_k},\bmu_{x_k}))\in
R^{\mathbb{F}_{\ourMinion}(\A)}$, so $\gamma:\X\rightarrow\mathbb{F}_{\ourMinion}(\A)$ is a homomorphism. 
\end{proof}
\noindent Our second goal is to prove the following.
\begin{thm}\label{CBLP_works_implies_minion_homo}
If $\CLAP$ solves $\PCSP(\A,\B)$ then there is a minion homomorphism from $\ourMinion$ to $\Pol(\A,\B)$.
\end{thm}
\begin{rem}
The proof of Theorem~\ref{CBLP_works_implies_minion_homo} proceeds essentially by establishing that
  the free structure $\mathbb{F}_\ourMinion(\A)$ has the $\CLAP$ condition as an
  instance of $\PCSP(\A,\B)$. However, some care is needed when handling Proposition~\ref{prop_2208_2504}, which only applies to \textit{finite} structures, while $\mathbb{F}_\ourMinion(\A)$ is not finite in general. 
To overcome this problem, we use a compactness argument tailored to our minion $\ourMinion$ discussed in Section~\ref{sec_compactness_argument}, which follows the ideas of~\cite{BBKO21}. 

We remark that the compactness argument for relational structures in the form stated in~\cite[Lemma A.6]{bgwz20} does not entirely fit our proof structure, as the element $(\be_1\bone_{\aleph_0}^T,\be_1)$ having the role of $\bar\bx$ in Definition~\ref{defn_CBLPC_2504_2212} does not belong to every induced substructure of  $\mathbb{F}_\ourMinion(\A)$. A different option would have been to use the general compactness argument known as the (uncountable version of the) \emph{compactness theorem of logic}~\cite{malcev1937untersuchungen}, that applies to all \emph{minion tests}\footnote{Cf.~Remark~\ref{rem_CLAP_is_a_minion_test}.} as derived in~\cite[Proposition~6]{cz23soda:minions} through~\cite{RorabaughTW17}.
\end{rem}
\begin{proof}[Proof of Theorem~\ref{CBLP_works_implies_minion_homo}]
Let $n=|A|$. For $D\in\N$, denote $\mathbb{F}_{{\ourMinion}_D}(\A)$ by
  $\mathbf{F}$ (where ${\ourMinion}_D$ is the sub-minion of $\ourMinion$ introduced in Section~\ref{sec_compactness_argument}). Hence, the domain of $\mathbf{F}$ is $\ourMinion_D^{(n)}$, which
  is finite. We claim that $\mathbf{F}$ has the $\CLAP$ condition as an instance
  of $\PCSP(\A,\B)$.

For each $R\in \sigma$ of arity $k$ and for each
  ${\btau}=((M_1,\bmu_1),\dots,(M_k,\bmu_k))\in {R}^{\mathbf{F}}$,
  take $(Q_{\btau},{\bdelta}_{\btau})\in \ourMinion_D^{(|R^\A|)}$ satisfying
  $(M_j,\bmu_j)=(Q_{\btau},{\bdelta}_{\btau})_{/\pi_j}$ $\forall j\in [k]$, where
  $\pi_j:R^\A\rightarrow A$ maps $\ba\in R^\A$ to its $j$-th coordinate; i.e., $M_j={{E}^{(R^\A,j)}}Q_{\btau}$ and $\bmu_j={{E}^{(R^\A,j)}}{\bdelta}_{\btau}$ $\forall j\in [k]$.
Given $R\in \sigma$ of arity $k$, consider the map
\begin{align*}
s^{R}:{R}^{\mathbf{F}}&\rightarrow\power({R}^\A)\setminus \{\emptyset\}\\
{\btau}&\mapsto \bigcup_{i\in\N}\supp(Q_{\btau} \be_i).
\end{align*}
Let us first check part $(I)$ of Definition~\ref{defn_CBLPC_2504_2212}.
Pick ${\btau}=((M_1,\bmu_1),\dots,(M_k,\bmu_k))\in {R}^{\mathbf{F}}$ and $\ba=(a_1,\dots,a_k)\in s^R({\btau})$. We have that $\ba\in\supp(Q_{\btau} \be_\alpha)$ for some $\alpha\in\N$, i.e., $\be_\ba^TQ_{\btau} \be_\alpha\neq 0$. Since $Q_{\btau}$ is skeletal, the set $L_{{\btau},\ba}=\{\ell\in\N:Q_{\btau} \be_\ell=\be_\ba\}$ is nonempty; let $\ell({\btau},\ba)\coloneqq\min(L_{{\btau},\ba})$. Consider the map
\begin{align*}
h_{{\btau},\ba}:\ourMinion_D^{(n)}&\rightarrow\mathbb{S}(A)\\
(\hat M,\hat\bmu)&\mapsto \hat M\be_{\ell({\btau},\ba)}.
\end{align*}
  We claim that $h_{{\btau},\ba}$ is a homomorphism from
  $\mathbf{F}$ to $\mathbb{S}(\A)$. Take $\tilde{R}\in\sigma$
  of arity $\tilde{k}$, and let
  $\tilde{\btau}=((\tilde{M}_1,\tilde{\bmu}_1),\dots,(\tilde{M}_{\tilde{k}},\tilde{\bmu}_{\tilde{k}}))\in\tilde{R}^{\mathbf{F}}$. Consider the pair $(Q_{\tilde{{\btau}}},{\bdelta}_{\tilde{{\btau}}})\in \ourMinion_D^{(|\tilde{R}^\A|)}$. We have that
\begin{align*}
h_{{\btau},\ba}(\tilde{{\btau}})&=(\tilde{M}_1\be_{\ell({\btau},\ba)},\dots,\tilde{M}_{\tilde{k}}\be_{\ell({\btau},\ba)})
=\left( {{E}^{(\tilde{R}^\A,1)}}Q_{\tilde{{\btau}}}\be_{\ell({\btau},\ba)},\dots,{{E}^{(\tilde{R}^\A,\tilde{k})}}Q_{\tilde{{\btau}}}\be_{\ell({\btau},\ba)}  \right).
\end{align*}
Since $Q_{\tilde{{\btau}}}\be_{\ell({\btau},\ba)}\in\mathbb{S}(\tilde{R}^\A)$, we deduce
  that $h_{{\btau},\ba}(\tilde{{\btau}})\in\tilde{R}^{\mathbb{S}(\A)}$, as wanted.
  Therefore, $h_{{\btau},\ba}$ is a homomorphism from $\mathbf{F}$ to $\mathbb{S}(\A)$. We now check that the requirements $1$ and $2$ in Definition~\ref{defn_CBLPC_2504_2212} are met. The former follows from
\begin{align*}
h_{{\btau},\ba}((M_i,\bmu_i))=M_i\be_{\ell({\btau},\ba)}={{E}^{(R^\A,i)}}Q_{\btau} \be_{\ell({\btau},\ba)}={{E}^{(R^\A,i)}} \be_\ba=\be_{a_i} && \forall i\in [k].
\end{align*}
To check the latter requirement, take $\tilde R\in\sigma$ of arity $\tilde k$
  and $\tilde
  {\btau}=((\tilde{M}_1,\tilde{\bmu}_1),\dots,(\tilde{M}_{\tilde{k}},\tilde{\bmu}_{\tilde{k}}))\in
  \tilde{R}^{\mathbf{F}}$, and consider $\bxi\coloneqq Q_{\tilde{{\btau}}}\be_{\ell({\btau},\ba)}$. Observe that
\begin{align*}
\bullet &&&h_{{\btau},\ba}((\tilde{M}_i,\tilde{\bmu}_i))=\tilde{M}_i \be_{\ell({\btau},\ba)} = {{E}^{(\tilde{R}^\A,i)}}Q_{\tilde{{\btau}}}\be_{\ell({\btau},\ba)}={{E}^{(\tilde{R}^\A,i)}}\bxi && \forall i\in[\tilde{k}]\\
\bullet &&& \supp(\bxi)=\supp(Q_{\tilde{{\btau}}}\be_{\ell({\btau},\ba)})\subseteq \bigcup_{i\in\N}\supp(Q_{\tilde{{\btau}}}\be_i)=s^{\tilde{R}}(\tilde{{\btau}}).
\end{align*}
We now check part $(II)$ of Definition~\ref{defn_CBLPC_2504_2212}. Take $\un$ as $\bar{R}$, and observe that 
\begin{align*}
  \un^{\mathbf{F}}=
\{(M,\bmu)\in \ourMinion_D^{(n)}:\exists (Q,{\bdelta})\in\ourMinion_D^{(n)} \mbox{ such that }M={{E}^{(\un^\A,1)}}Q, \bmu={{E}^{(\un^\A,1)}}{\bdelta}\}=
\ourMinion_D^{(n)},
\end{align*}
where we have used that ${E}^{(\un^\A,1)}=I_n$. Consider the element $\bar{\tau}=(\be_1\bone_{\aleph_0}^T,\be_1)\in \ourMinion_D^{(n)}=\un^{\mathbf{F}}$. Using again that ${E}^{(\un^\A,1)}=I_n$, we see that $(Q_{\bar{\tau}},{\bdelta}_{\bar{\tau}})=\bar{\tau}$. We obtain
\begin{align*}
s^{\un}(\bar{\tau})=\bigcup_{i\in\N}\supp(\be_1\bone_{\aleph_0}^T\be_i)=
\bigcup_{i\in\N}\supp(\be_1)=\{1\}.
\end{align*}
Hence, we pick $\bar{a}=1$. Notice that
\begin{align*}
\ell(\bar{\tau},\bar{a})=\min\{\ell\in\N:\be_1\bone_{\aleph_0}^T\be_\ell=\be_1\}=
\min\{\ell\in\N:\be_1=\be_1\}=\min\N= 1.
\end{align*}
Consider the function
\begin{align*}
g:\ourMinion_D^{(n)}&\rightarrow\mathbb{Z}(A)\\
(\hat{M},\hat{\bmu})&\mapsto \hat{\bmu}.
\end{align*}
Following the same procedure as for $h_{{\btau},\ba}$, we easily check that $g$ is a
homomorphism from $\mathbf{F}$ to $\mathbb{Z}(\A)$. We now
verify that condition $1'$ of Definition~\ref{defn_CBLPC_2504_2212} is
satisfied. Given $\tilde{R}\in\sigma$ of arity $\tilde{k}$ and
$\tilde{{\btau}}=((\tilde{M}_1,\tilde{\bmu}_1),\dots,(\tilde{M}_{\tilde{k}},\tilde{\bmu}_{\tilde{k}}))\in\tilde{R}^{\mathbf{F}}$, let $\bxi\coloneqq Q_{\tilde{{\btau}}}\be_1\in\mathbb{S}(\tilde{R}^\A)$ and $\bzeta\coloneqq {\bdelta}_{\tilde{{\btau}}}\in\mathbb{Z}(\tilde{R}^\A)$. Then, given $i\in [\tilde{k}]$,
\begin{itemize}
\item[$\ast$] $\displaystyle h_{\bar{\tau},\bar{a}}((\tilde{M}_i,\tilde{\bmu}_i))=\tilde{M}_i\be_{\ell(\bar{\tau},\bar{a})}=
\tilde{M}_i\be_1=
{{E}^{(\tilde{R}^\A,i)}}Q_{\tilde{{\btau}}}\be_1=
{{E}^{(\tilde{R}^\A,i)}}\bxi$;
\item[$\ast$] $\displaystyle g((\tilde{M}_i,\tilde{\bmu}_i))=\tilde{\bmu}_i={{E}^{(\tilde{R}^\A,i)}}{\bdelta}_{\tilde{{\btau}}}={{E}^{(\tilde{R}^\A,i)}}\bzeta$;
\item[$\ast$] $\displaystyle \supp(\bzeta)=\supp({\bdelta}_{\tilde{{\btau}}})\subseteq \supp(Q_{\tilde{{\btau}}}\be_1)=\supp(\bxi)\subseteq \bigcup_{i\in \N}\supp(Q_{\tilde{{\btau}}}\be_i)=s^{\tilde{R}}(\tilde{{\btau}})$
\end{itemize}
where, for the first inclusion in the third line, we have used $(c_4)$ in Definition~\ref{def:minion}.

It follows that $\mathbf{F}$ has the $\CLAP$ condition as an
instance of $\PCSP(\A,\B)$, as claimed. Then, Proposition~\ref{prop_2208_2504}
implies that $\CLAP$ accepts $\textbf{F}$. Since, by hypothesis, $\CLAP$ solves
$\PCSP(\A,\B)$, we deduce that $\mathbb{F}_{\ourMinion_D}(\A)=\mathbf{F}\rightarrow\B$. By
Lemma~\ref{lem:free}, there exists a minion homomorphism from $\ourMinion_D$ to $\Pol(\A,\B)$. Finally, since the set of polymorphisms of $(\A,\B)$ of arity $L$ is finite for every $L\in\N$, Proposition~\ref{prop_compactness_argument} allows us to conclude that there exists a minion homomorphism from $\ourMinion$ to $\Pol(\A,\B)$.
\end{proof}

\begin{rem}
\label{rem_CLAP_is_a_minion_test}
It follows from the proofs of Theorems~\ref{minion_homo_implies_CBLP_works} and~\ref{CBLP_works_implies_minion_homo} that $\CLAP$ fits within the framework of \emph{minion tests} recently introduced in~\cite{cz23soda:minions}. More precisely, $\CLAP=\operatorname{Test}_{\ourMinion}$, which means that, for two $\sigma$-structures $\X$ and $\A$, $\CLAP(\X,\A)$ accepts if and only if $\X\to\mathbb{F}_{\ourMinion}(\A)$. Additionally, it follows from~\cite{cz23soda:minions} that $\CLAP$ is a \emph{conic} minion test, which essentially means that one can build a progressively tighter hierarchy of relaxations based on $\CLAP$ whose $k$-th level correctly classifies all instances of size $k$. 
\end{rem}

\section{H-symmetric polymorphisms}
\label{sec_H_symmetric}

This section contains the proof of Theorem~\ref{thm:Hsymm}. We remark that the machinery developed here can be extended to the more general setting of $H$-block-symmetric polymorphisms, at the only cost of dealing with a more cumbersome notation. This is done in Appendix~\ref{sec:block_sym_pol} and results in Theorem~\ref{thm_CLAP_works_H_block_sym} -- a slightly stronger version of Theorem~\ref{thm:Hsymm}.

We shall need two helpful lemmas.
The first lemma shows a property of $H$-symmetric functions that will be useful in the proof of Theorem~\ref{thm:Hsymm}. Throughout this section, without loss of generality, we consider $A=[n]$.
\begin{lem}
\label{lem_1347_1805}
Let $f:A^L\rightarrow B$ be $H$-symmetric for some $m\times n$ tie matrix $H$, with $m\in\N$. Consider two maps $\pi,\tilde{\pi}:[L]\rightarrow [n]$ such that ${P_\pi}\bone_L={P_{\tilde{\pi}}}\bone_L$ and the vector ${P_\pi}\bone_L$ is $H$-tieless. Then
\begin{align*}
f_{/\pi}(1,\dots,n)=f_{/\tilde{\pi}}(1,\dots,n).
\end{align*}
\end{lem}
\begin{proof}
For ${a}\in [n]$, we have
\begin{align*}
|\pi^{-1}({a})|=\sum_{i\in [L]}({P_\pi})_{{a}i}=\sum_{i\in [L]}\be_{a}^T{P_\pi}\be_i=
\be_{a}^T{P_\pi}\bone_L
=
\be_{a}^T{P_{\tilde\pi}}\bone_L
=
|\tilde\pi^{-1}({a})|.
\end{align*}
Hence, we can consider bijections $\varphi_{a}:\pi^{-1}({a})\rightarrow\tilde{\pi}^{-1}({a})$ for each ${a}\in [n]$. Clearly, their union 
\begin{align*}
\varphi=\bigcup_{{a}\in [n]}\varphi_{a}:[L]\rightarrow [L]
\end{align*}
 is also a bijection. For each $i\in [L]$, we have
\begin{align*}
(\tilde{\pi}\circ \varphi)(i)=\tilde{\pi}(\varphi(i))=
\tilde{\pi}(\varphi_{\pi(i)}(i))=\pi(i) 
\end{align*}
and, hence, $\tilde{\pi}\circ\varphi=\pi$.
Let $\tilde{\ba}=(\tilde{\pi}(1),\dots,\tilde{\pi}(L))$. Notice that, for each $a\in [n]$,
\begin{align*}
\be_a^T\tilde{\ba}^\#=|\{i\in [L]:\tilde{\pi}(i)=a\}|=\be_a^T{P_{\tilde{\pi}}}\bone_L
\end{align*} 
and, therefore, $\tilde{\ba}^\#={P_{\tilde{\pi}}}\bone_L={P_\pi}\bone_L$, which is $H$-tieless.
Using that $f$ is $H$-symmetric, we find
\begin{align*}
f_{/\tilde{\pi}}(1,\dots,n)=f(\tilde{\ba})=f_{/\varphi}(\tilde{\ba})=(f_{/\varphi})_{/\tilde{\pi}}(1,\dots,n)=
f_{/\tilde{\pi}\circ\varphi}(1,\dots,n)=f_{/\pi}(1,\dots,n),
\end{align*}
as required.
\end{proof}

One intriguing property of skeletal matrices is that they can be simultaneously reduced to $H$-tieless vectors, in the sense of the next lemma. We say that a vector is \emph{finitely supported} if it only has
a finite number of nonzero entries.
\begin{lem}[Tiebreak Lemma]
\label{lemma_tie_terminator}
For $k,p,m\in\N$, let $M_1,\dots,M_k\in\Q^{p,\aleph_0}$ be skeletal matrices, and let $H$ be an $m\times p$ tie matrix. Then there exists a stochastic finitely supported vector $\bv\in\Q^{\aleph_0}$ with $\be_1^T\bv>0$ such that $M_j\bv$ is $H$-tieless for any $j\in[k]$. 
\end{lem}
\begin{proof}
Let $\Omega$ be the set of rational stochastic finitely supported vectors of size ${\aleph_0}$ whose first entry is nonzero, and consider the map
\begin{align*}
f:\Omega&\rightarrow \N_0\\
\hat \bv&\mapsto \sum_{j\in [k]}|\{(i,i')\in [m]^2:i\neq i' \mbox{ and } \be_i^THM_j\hat{\bv}=\be_{i'}^THM_j\hat{\bv}\neq 0\}|.
\end{align*}
In other words, $f(\hat{\bv})$ counts the total number of ties in the set of vectors $\{HM_j\hat{\bv}:j\in[k]\}$.
Let $\bv$ attain the minimum of $f$ over $\Omega$. If $f(\bv)=0$, we are done.
  Otherwise, let $j\in [k]$, $i,i'\in [m]$ be such that $i\neq i'$ and
  $\be_i^THM_j\bv=\be_{i'}^THM_j\bv\neq 0$. From $\be_i^THM_j\bv\neq 0$, we see that
  $\exists \beta\in [p]$ such that $\be_i^TH\be_\beta\neq 0$ and $\be_\beta^TM_j\bv\neq
  0$. In particular, we have $\be_\beta^TM_j\neq \bzero_{\aleph_0}^T$; since $M_j$ is skeletal, this implies that $M_j\be_\alpha=\be_\beta$ for some $\alpha\in \N$. For $\epsilon\in\Q$, $0<\epsilon< 1$, consider the vector $\bv_\epsilon=(1-\epsilon)\bv+\epsilon \be_\alpha$. Observe that $\bv_\epsilon\in\Omega$. For $g\in [k]$, we have $HM_g\bv_\epsilon=(1-\epsilon)HM_g\bv+\epsilon HM_g\be_\alpha$. By choosing $\epsilon$ sufficiently small, we can assume that, for each $g\in [k]$, $HM_g\bv_\epsilon$ does not have new ties other than those in $HM_g\bv$. Moreover, 
\begin{align*}
HM_j\bv_\epsilon&= (1-\epsilon)HM_j\bv+\epsilon HM_j \be_\alpha= (1-\epsilon)HM_j\bv + \epsilon H\be_\beta
\intertext{and, hence,}
\be_i^THM_j\bv_\epsilon &=(1-\epsilon)\be_i^THM_j\bv+\epsilon \be_i^TH\be_\beta
=
(1-\epsilon)\be_{i'}^THM_j\bv+\epsilon \be_i^TH\be_\beta\\
&\neq
(1-\epsilon)\be_{i'}^THM_j\bv+\epsilon \be_{i'}^TH\be_\beta
=
\be_{i'}^THM_j\bv_\epsilon,
\end{align*}
where the disequality follows from $\be_i^TH\be_\beta\neq 0$ and from the fact that $H\be_\beta$ is a tieless vector by the definition of tie matrix. We conclude that $f(\bv_\epsilon)<f(\bv)$, which contradicts our assumption.
\end{proof}

\begin{thm*}[Theorem~\ref{thm:Hsymm} restated]
Let $(\A,\B)$ be a PCSP template and suppose $\Pol(\A,\B)$ contains
  $H$-symmetric operations of arbitrarily large arity for some $m\times |A|$ tie
  matrix $H$, $m\in\N$. Then there exists a minion homomorphism from
  $\ourMinion$ to $\Pol(\A,\B)$. 
\end{thm*}
\begin{rem}
Before proving Theorem~\ref{thm:Hsymm}, we provide some intuition on the construction of the minion homomorphism whose existence shall establish the result. First, one fixes an $H$-symmetric polymorphism $f$. Then, the image of an $L$-ary element $(M,\bmu)$ of $\ourMinion$ under the homomorphism is the function that $(i)$ takes a tuple $(a_1,\dots,a_L)$ of variables in $A$ as input, $(ii)$ \emph{deforms} the tuple by changing the frequency of each variable according to the information carried by $M$ and $\bmu$, and $(iii)$ returns as output the evaluation of $f$ on the deformed tuple. The deformation in step $(ii)$ is encoded by the map $\varphi$ defined in~\eqref{expr_P_phi}. Essentially, $\varphi$ decides what frequency to assign to a variable $a_i$ on the basis of the weight of $i$ in the probability distribution $M\bv$ -- where $\bv$ is the \emph{tie-breaking} vector from Lemma~\ref{lemma_tie_terminator}. 
 The integer distribution $\bmu$ is also taken into account by $\varphi$, and its role is essentially to fill the gap between the size of the deformed tuple obtained above and the arity of $f$. If $\Pol(\A,\B)$ is rich enough to provide $H$-symmetric polymorphisms of whichever arity we need, $\bmu$ is inessential (cf.~Remark~\ref{rem_H_sym_all_arities}).  
\end{rem}

\begin{proof}[Proof of Theorem~\ref{thm:Hsymm}]
For $D\in\N$, consider the subminion $\ourMinion_D$ of $\ourMinion$ described in Section~\ref{sec_compactness_argument}. Observe that $S=\{M:(M,\bmu)\in\ourMinion_D^{(n)}\}$ is a finite set of skeletal matrices. Therefore, we can apply the Tiebreak Lemma~\ref{lemma_tie_terminator} to find a stochastic finitely supported vector ${\bv}\in \Q^{\aleph_0}$ with $\be_1^T{\bv}>0$ such that $M{\bv}$ is $H$-tieless for any $M\in S$. Since ${\bv}$ is finitely supported, we can find $N'\in\N$ such that $N'{\bv}$ has integer entries. Let $\sigma_1^H$ denote the largest singular value of $H$ -- i.e., the square root of the largest eigenvalue of $H^TH$. Set $N=2\lceil\sigma_1^H+1\rceil D^2N'$, and let $f$ be an $H$-symmetric polymorphism of arity ${c}\geq N^2$. Write ${c}=N{\alpha}+{\beta}$ with ${\alpha},{\beta}\in \N_0$, ${\beta}\leq N-1$. Note that $N^2\leq {c}= N{\alpha}+{\beta}\leq N{\alpha}+N-1<N({\alpha}+1)$, so $N<{\alpha}+1$ and, hence, ${\beta}<{\alpha}$.

Consider the function 
\begin{align*}
\xi_D:\ourMinion_D\rightarrow\Pol(\A,\B)
\end{align*} 
defined as follows. Given $L\in\N$ and $(M,\bmu)\in\ourMinion_D^{(L)}$, take the map $\varphi:[{c}]\rightarrow [L]$ such that the corresponding $L\times {c}$ matrix ${P_\varphi}$ is 
\begin{align}
\label{expr_P_phi}
{P_\varphi}=
\begin{pmatrix}
\bone^T_{\be_1^T({\alpha}NM{\bv}+{\beta}\bmu)} & \bzero^T & \dots & \bzero^T \\ 
\bzero^T & \bone^T_{\be_2^T({\alpha}NM{\bv}+{\beta}\bmu)} & \dots & \bzero^T \\ 
\vdots & \vdots & \ddots & \vdots \\ 
\bzero^T & \bzero^T & \dots & \bone^T_{\be_L^T({\alpha}NM{\bv}+{\beta}\bmu)}
\end{pmatrix}.
\end{align}
To verify that~\eqref{expr_P_phi} is well defined, observe first that
\begin{align*}
\sum_{i=1}^L \be_i^T({\alpha}NM{\bv}+{\beta}\bmu)=\bone_L^T({\alpha}NM{\bv}+{\beta}\bmu)={\alpha}N\bone_L^TM{\bv}+{\beta}\bone_L^T\bmu={\alpha}N\bone_{\aleph_0}^T{\bv}+{\beta}={\alpha}N+{\beta}={c}.
\end{align*}
Moreover, for each $i\in [L]$, $\be_i^T({\alpha}NM{\bv}+{\beta}\bmu)=\be_i^T(2{\alpha}\lceil\sigma_1^H+1\rceil D(DM)(N'{\bv})+{\beta}\bmu)$ is an integer. If $\be_i^T({\alpha}NM{\bv}+{\beta}\bmu)$ was negative, then $\be_i^T\bmu<0$. By the requirement $(c_4)$ in Definition~\ref{def:minion}, this would imply that $\be_i^TM\be_1>0$ and, hence, $0<\be_i^TM\be_1\be_1^T{\bv}\leq \be_i^TM{\bv}$. As a consequence, $\be_i^T(DM)(N'{\bv})\geq 1$ so that
\begin{align*}
\be_i^T({\alpha}NM{\bv}+{\beta}\bmu)=2{\alpha}\lceil\sigma_1^H+1\rceil D\be_i^T(DM)(N'{\bv})+{\beta}\be_i^T\bmu\geq 2{\alpha}\lceil\sigma_1^H+1\rceil D+{\beta}\be_i^T\bmu\geq {\alpha}D-{\beta}D>0,
\end{align*}
which is a contradiction. In conclusion, the numbers $\be_i^T({\alpha}NM{\bv}+{\beta}\bmu)$ are nonnegative integers summing up to ${c}$, so~\eqref{expr_P_phi} is well defined.

We define $\xi_D((M,\bmu))\coloneqq f_{/\varphi}$. Clearly, $\xi_D((M,\bmu))\in \Pol(\A,\B)$. We claim that the map $\xi_D$ is a minion homomorphism. It is straightforward to check that $\xi_D$ preserves arities so, to conclude, we need to show that it also preserves minors. Take $L'\in\N$ and choose a map $\pi:[L]\rightarrow [L']$. Letting $\tilde{\varphi}:[{c}]\rightarrow [L']$ be the map corresponding to the matrix
\begin{align*}
{P_{\tilde\varphi}}=
\begin{pmatrix}
\bone^T_{\be_1^T({\alpha}N{P_\pi}M{\bv}+{\beta}{P_\pi}\bmu)} & \bzero^T & \dots & \bzero^T \\ 
\bzero^T & \bone^T_{\be_2^T({\alpha}N{P_\pi}M{\bv}+{\beta}{P_\pi}\bmu)} & \dots & \bzero^T \\ 
\vdots & \vdots & \ddots & \vdots \\ 
\bzero^T & \bzero^T & \dots & \bone^T_{\be_{L'}^T({\alpha}N{P_\pi}M{\bv}+{\beta}{P_\pi}\bmu)}
\end{pmatrix}, 
\end{align*}
we see that $\xi_D((M,\bmu)_{/\pi})=f_{/\tilde{\varphi}}$. Moreover, $\xi_D((M,\bmu))_{/\pi}=(f_{/\varphi})_{/\pi}=f_{/\pi\circ\varphi}$, where $\varphi$ corresponds to the matrix ${P_\varphi}$ in~\eqref{expr_P_phi}. Take ${\ba}=(a_1,\dots,a_{L'})\in A^{L'}$, and consider the map 
\begin{align*}
\pi_{\ba}:[L']&\rightarrow [n]\\
i&\mapsto a_i.
\end{align*}
Observe that 
\begin{align}
\label{eqn_f_phi_phi_tilde}
f_{/\tilde{\varphi}}(\ba)&=&(f_{/\tilde{\varphi}})_{/\pi_{\ba}}(1,\dots,n)&=&f_{/\pi_{\ba}\circ \tilde{\varphi}}(1,\dots,n)&&\mbox{and, similarly,}\notag\\
f_{/\pi\circ\varphi}(\ba)&=&(f_{/\pi\circ\varphi})_{/\pi_{\ba}}(1,\dots,n)&=&f_{/\pi_{\ba}\circ\pi\circ\varphi}(1,\dots,n).
\end{align}
Notice that
\begin{align*}
{P_{\pi_{\ba}\circ \tilde{\varphi}}}\bone_{c}&=
{P_{\pi_{\ba}}}{P_{\tilde{\varphi}}}\bone_{c}=
{P_{\pi_{\ba}}}({\alpha}N{P_\pi}M{\bv}+{\beta}{P_\pi}\bmu)\\
&=
{P_{\pi_{\ba}}}{P_\pi}({\alpha}NM{\bv}+{\beta}\bmu)=
{P_{\pi_{\ba}}}{P_\pi}{P_\varphi}\bone_{c}={P_{\pi_{\ba}\circ\pi\circ\varphi}}\bone_{c}.
\end{align*}
We claim that the vector ${P_{\pi_{\ba}\circ \tilde{\varphi}}}\bone_{c}$ is $H$-tieless. Let ${\bu}=(u_i)=H{P_{\pi_{\ba}\circ \tilde{\varphi}}}\bone_{c}$; the claim is equivalent to ${\bu}$ being tieless. Let ${\bw}=(w_i)={\alpha}NH{P_{\pi_{\ba}\circ\pi}}M{\bv}$ and ${\bz}=(z_i)={\beta}H{P_{\pi_{\ba}\circ\pi}}\bmu$, so that ${\bu}={\bw}+{\bz}$. Choose $i,i'\in [m]$ such that $i\neq i'$ and $u_i\neq 0$. We need to show that $u_i\neq u_{i'}$. Suppose $w_i=0$. We can write $H^T\be_i=\sum_{g\in G}\lambda_g \be_g$ for $G=\supp(H^T\be_i)$, where each $\lambda_g$ is a positive integer (note that $G\neq\emptyset$ since, otherwise, $H^T\be_i=\bzero_n$, which would imply $u_i=0$). Let $F=(\pi_{\ba}\circ \pi)^{-1}(G)$. From $w_i=0$, we obtain
\begin{align*}
&&0=\be_i^TH{P_{\pi_{\ba}\circ\pi}}M{\bv}=(H^T\be_i)^T{P_{\pi_{\ba}\circ\pi}}M{\bv}
&=&
\sum_{g\in G}\lambda_g\be_g^T{P_{\pi_{\ba}\circ\pi}}M{\bv}
=
\sum_{g\in G}\lambda_g\sum_{j\in (\pi_{\ba}\circ\pi)^{-1}(g)}\be_j^TM{\bv}
\end{align*}
and, hence, the following chain of implications holds:
\begin{align*}
&&\begin{array}{lllll}
&\displaystyle
0=\sum_{g\in G}\sum_{j\in (\pi_{\ba}\circ\pi)^{-1}(g)}\be_j^TM{\bv}
=
\sum_{j\in F}\be_j^TM{\bv}
&\displaystyle
\hspace{.2cm}\Rightarrow\hspace{.2cm}
&\displaystyle
\be_j^TM{\bv}=0&{\forall j\in F}
\\[7pt]
\Rightarrow\hspace{.2cm}
&\displaystyle
\be_j^TM\be_1=0\hspace{1cm}{\forall j\in F}
&\hspace{.2cm}\Rightarrow\hspace{.2cm}
&\displaystyle
\be_j^T\bmu=0&{\forall j\in F}
\end{array}
\end{align*}
(where the second implication follows from $\be_1^T{\bv}>0$, and the third follows from $(c_4)$ in Definition~\ref{def:minion}). Hence,
\begin{align*}
z_i={\beta}\be_i^TH{P_{\pi_{\ba}\circ\pi}}\bmu
=
{\beta}\sum_{g\in G}\lambda_g\be_g^T{P_{\pi_{\ba}\circ\pi}}\bmu
=
{\beta}\sum_{g\in G}\lambda_g\sum_{j\in (\pi_{\ba}\circ\pi)^{-1}(g)}\be_j^T\bmu=0,
\end{align*}
so that $u_i=w_i+z_i=0$, a contradiction. Hence, $w_i>0$. Observe that $(M_{/\pi_{\ba}\circ\pi},\bmu_{/\pi_{\ba}\circ\pi})\in \ourMinion_D^{(n)}$ and, hence, $M_{/\pi_{\ba}\circ\pi}\in S$. By the choice of ${\bv}$, this implies that the vector ${P_{\pi_{\ba}\circ \pi}}M{\bv}=M_{/\pi_{\ba}\circ \pi}{\bv}$ is $H$-tieless; i.e., $H{P_{\pi_{\ba}\circ \pi}}M{\bv}$ is tieless. It follows that the vector $H{P_{\pi_{\ba}\circ \pi}}(DM)(N'{\bv})=\frac{1}{2{\alpha}\lceil\sigma_1^H+1\rceil D}{\bw}$ is also tieless; being it entrywise integer, and since $\frac{1}{2{\alpha}\lceil\sigma_1^H+1\rceil D}w_i>0$, we obtain
\begin{align*}
\left|\frac{1}{2{\alpha}\lceil\sigma_1^H+1\rceil D}w_i-\frac{1}{2{\alpha}\lceil\sigma_1^H+1\rceil D}w_{i'}\right|\geq 1&&\mbox{that yields}&&|w_i-w_{i'}|\geq 2{\alpha}\lceil\sigma_1^H+1\rceil D.
\end{align*}
Denote the $\ell_1$-norm\ and the $\ell_2$-norm of a vector by $\|\cdot\|_1$ and
$\|\cdot\|_2$, respectively. Recall that the largest singular value of a matrix is its spectral operator norm -- i.e., $\sigma_1^H=\max_{\bzero\neq\bx\in \R^{n}}\frac{\|H\bx\|_2}{\|\bx\|_2}$ (see~\cite{Horn2012matrix}). In particular, $\|H\bx\|_2\leq\sigma_1^H\|\bx\|_2$ for each vector $\bx$ of size $n$.
Using the Cauchy-Schwarz inequality and the fact that the $\ell_1$-norm of a vector is greater than or equal to its $\ell_2$-norm, we find
\begin{align*}
|z_i-z_{i'}|&={\beta}|(\be_i-\be_{i'})^TH{P_{\pi_{\ba}\circ\pi}}\bmu|\leq 
{\beta}\|\be_i-\be_{i'}\|_2\|H{P_{\pi_{\ba}\circ\pi}}\bmu\|_2
\leq
{\beta}\|\be_i-\be_{i'}\|_2\sigma_1^H\|{P_{\pi_{\ba}\circ\pi}}\bmu\|_2\\
&\leq 
{\beta}\|\be_i-\be_{i'}\|_1\lceil\sigma_1^H+1\rceil\|{P_{\pi_{\ba}\circ\pi}}\bmu\|_1
=
2{\beta}\lceil\sigma_1^H+1\rceil \bone_n^T|{P_{\pi_{\ba}\circ\pi}}\bmu|
\leq
2{\beta}\lceil\sigma_1^H+1\rceil \bone_n^T{P_{\pi_{\ba}\circ\pi}}|\bmu|\\
&=2{\beta}\lceil\sigma_1^H+1\rceil \bone_L^T|\bmu|
\leq 
2{\beta}\lceil\sigma_1^H+1\rceil D < 2{\alpha}\lceil\sigma_1^H+1\rceil D.
\end{align*}
We conclude the proof of the claim by noting that 
\begin{align*}
|u_i-u_{i'}|=|(w_i-w_{i'})-(z_{i'}-z_i)|\geq |w_i-w_{i'}|-|z_i-z_{i'}|>2{\alpha}\lceil\sigma_1^H+1\rceil D-2{\alpha} \lceil\sigma_1^H+1\rceil D=0,
\end{align*}
which implies $u_i\neq u_{i'}$. As a consequence, the vector ${P_{\pi_{\ba}\circ
\tilde{\varphi}}}\bone_{c}$ is $H$-tieless. We can then apply Lemma~\ref{lem_1347_1805} to conclude that $f_{/\pi_{\ba}\circ
\tilde{\varphi}}(1,\dots,n)=f_{/\pi_{\ba}\circ\pi\circ\varphi}(1,\dots,n)$. Hence, by~\eqref{eqn_f_phi_phi_tilde}, $f_{/\tilde{\varphi}}=f_{/\pi\circ\varphi}$. Therefore,
$\xi_D((M,\bmu)_{/\pi})=\xi_D((M,\bmu))_{/\pi}$, as required. It follows that $\xi_D$ is a minion homomorphism. 

Since the set of polymorphisms of $(\A,\B)$ of arity $L$ is finite for every $L\in\N$, we can apply Proposition~\ref{prop_compactness_argument} to conclude that there exists a minion homomorphism $\zeta:\ourMinion\rightarrow\Pol(\A,\B)$.
\end{proof}
\begin{rem}
\label{rem_H_sym_all_arities}
If $\Pol(\A,\B)$ contains $H$-symmetric operations of \emph{all} arities -- as it
  happens for the PCSP template $(\A,\B)$ from Example~\ref{ex:new1}, cf.~Example~\ref{ex:new2}
  -- the $\AIP$ part of $\CLAP$ is not required. Indeed, in that case, we can
  choose $f$ in the proof of Theorem~\ref{thm:Hsymm} to be an $H$-symmetric
  polymorphism of arity $c=N^2$, which implies $\beta=0$. Therefore, the affine
  vector $\bmu$ does not have any role in the definition of ${P_\varphi}$
  in~\eqref{expr_P_phi}, nor in the definition of the minion homomorphism
  $\xi_D$. It follows that, under this stronger hypothesis, $\Pol(\A,\B)$ admits
  a minion homomorphism from a minion $\hat\ourMinion$ whose $L$-ary elements
  are matrices in $\Q^{L, \aleph_0}$ satisfying the requirements
  $(c_1),(c_2),(c_5),(c_6)$ of Definition~\ref{def:minion}; notice that the
  projection $(M,\bmu)\mapsto M$ yields a natural minion homomorphism from
  $\ourMinion$ to $\hat\ourMinion$. The proofs of
  Theorems~\ref{minion_homo_implies_CBLP_works}
  and~\ref{CBLP_works_implies_minion_homo} can be straightforwardly modified to
  show that $\hat{\ourMinion}$ captures the power of the algorithm $\CBLP$ --
  i.e., the simplified version of $\CLAP$ that does not run $\BLP+\AIP$ at the
  end (cf.~the discussion in Section~\ref{sec:clap}).
\end{rem}

\section*{Acknowledgements}

We would like to thank the anonymous referees of both
the conference~\cite{Ciardo22:soda} and this full version of the paper.

\appendix

\section{Existing relaxations for PCSPs}
\label{app:relaxations}

Every CSP can be equivalently expressed as a $0$--$1$ integer program in a
standard way. 

If the variables are allowed to take values in $[0,1]$, we obtain the so-called
basic linear programming relaxation ($\BLP$)~\cite{Kun12:itcs}. This naturally extends
to PCSPs~\cite{BBKO21}, as we describe in Appendix~\ref{subsec:blp}.

If the variables are allowed to take integer values,
we obtain the so-called basic affine integer programming
relaxation ($\AIP$)~\cite{BG21}, studied in detail in~\cite{BBKO21}, as we describe in
Appendix~\ref{subsec:aip}.

A combination of the two relaxations, called the $\BLP+\AIP$ relaxation, was proposed in~\cite{bgwz20} and its power
characterised in~\cite{bgwz20}, as we describe in Appendix~\ref{subsec:blp-aip}.

Let $(\A,\B)$ be a PCSP template with signature $\sigma$ and let $\X$ be an
instance of $\PCSP(\A,\B)$. In all three relaxations described below, we assume without loss of generality that $\sigma$
contains a unary symbol $\un$ such that $\un^\X=X$, $\un^\A=A$, and $\un^\B=B$.
If this is not the case, the signature and the instance can be extended without
changing the set of solutions.

\subsection{BLP}\label{subsec:blp}
The \emph{basic linear programming relaxation} ($\BLP$) of $\X$, denoted by $\BLP(\X,\A)$,
is defined as follows.\footnote{The definition does not depend on $\B$ and is
the same as the $\BLP$ of an instance $\X$ of $\CSP(\A)$; the same holds for $\AIP$ and $\BLP+\AIP$.} The variables are
$\lambda_{\bx,R}(\ba)$ for every $R\in\sigma$, $\bx\in R^\X$,
and $\ba\in R^\A$, and the constraints are given in Figure~\ref{fig:blp}.
\begin{figure}[hbt]
\begin{align*}
  0\ \leq\ \lambda_{\bx,R}(\ba)\ &\leq\ 1 & \forall R\in\sigma, \forall \bx\in R^\X, \forall \ba\in R^\A\\
  \sum_{\ba\in R^\A} \lambda_{\bx,R}(\ba)\ &=\ 1 & \hspace*{3cm} \forall R\in\sigma, \forall \bx\in R^\X\\
  \sum_{\ba\in R^\A,a_i=a} \lambda_{\bx,R}(\ba)\ &=\ \lambda_{x_i,\un}(a) & \hspace*{3cm} \forall R\in\sigma, \forall \bx\in R^\X, \forall a\in A,\forall i\in [\ar(R)]
\end{align*}
\caption{Definition of $\BLP(\X,\A)$.}
\label{fig:blp}
\end{figure}

\noindent We say that $\BLP(\X,\A)$ accepts if the LP in Figure~\ref{fig:blp} is
feasible, and rejects otherwise. By construction, if $\X\to\A$ then
$\BLP(\X,\A)$ accepts. We say that $\BLP$ \emph{solves} $\PCSP(\A,\B)$ if for
every instance $\X$ accepted by $\BLP(\X,\A)$ we have $\X\to\B$. 

We denote by $\Qconv$ the minion of stochastic vectors
on $\Q$ with the minor
operation defined as in Section~\ref{sec:minion}; i.e., if $\bq\in\Qconv^{(L)}$
and $\pi:[L]\to[L']$, then $\bq_{/\pi}={P_\pi}\bq$, where $P_\pi$ is the $L'\times
L$ matrix whose $(i,j)$-th entry is $1$ if $\pi(j)=i$, and $0$ otherwise.

An $L$-ary operation $f:A^L\to B$ is called \emph{symmetric} if
$f(a_1,\ldots,a_L)=f(a_{\pi(1)},\ldots,a_{\pi(L)})$ for every $a_1,\ldots,a_L\in
A$ and every permutation $\pi:[L]\to[L]$.

The power of $\BLP$ for PCSPs is characterised in the following result.

\begin{thm}[\cite{BBKO21}]\label{thm:blp}
Let $(\A,\B)$ be a PCSP template. The following are equivalent:
  \begin{enumerate}
    \item[(1)] $\BLP$ solves $\PCSP(\A,\B)$.
    \item[(2)] $\Pol(\A,\B)$ admits a minion homomorphism from $\Qconv$.
    \item[(3)] $\Pol(\A,\B)$ contains symmetric operations of all arities.
  \end{enumerate}
\end{thm}

\subsection{AIP}\label{subsec:aip}
The \emph{basic affine integer programming relaxation} ($\AIP$) of $\X$, denoted by
$\AIP(\X,\A)$, is defined as follows. The variables are
$\tau_{\bx,R}(\ba)$ for every
$R\in\sigma$, $\bx\in R^\X$, and $\ba\in R^\A$, and the constraints are given in
Figure~\ref{fig:aip}.
\begin{figure}[hbt]
\begin{align*}
  \tau_{\bx,R}(\ba)\ &\in\ \Z & \forall R\in\sigma, \forall \bx\in R^\X, \forall \ba\in R^\A\\
  \sum_{\ba\in R^\A} \tau_{\bx,R}(\ba)\ &=\ 1 & \hspace*{3cm} \forall R\in\sigma, \forall \bx\in R^\X\\
  \sum_{\ba\in R^\A,a_i=a} \tau_{\bx,R}(\ba)\ &=\ \tau_{x_i,\un}(a) & \hspace*{3cm} \forall R\in\sigma, \forall \bx\in R^\X, \forall a\in A, \forall i\in [\ar(R)]
\end{align*}
\caption{Definition of $\AIP(\X,\A)$.}
\label{fig:aip}
\end{figure}

\noindent We say that $\AIP(\X,\A)$ accepts if the affine program in
Figure~\ref{fig:aip} is feasible, and rejects otherwise. By construction, if
$\X\to\A$ then $\AIP(\X,\A)$ accepts. We say that $\AIP$ \emph{solves}
$\PCSP(\A,\B)$ if for every instance $\X$ accepted by $\AIP(\X,\A)$ we have
$\X\to\B$. 

We denote by $\Zaff$ the minion of affine vectors
on $\Z$ with the minor operation defined as in
Section~\ref{sec:minion}; i.e., if $\bz\in\Zaff^{(L)}$
and $\pi:[L]\to[L']$, then $\bz_{/\pi}={P_\pi}\bz$, where ${P_\pi}$ is the $L'\times
L$ matrix whose $(i,j)$-th entry is $1$ if $\pi(j)=i$, and $0$ otherwise.

A $(2L+1)$-ary operation $f:A^{2L+1}\to B$ is called \emph{alternating} if
$f(a_1,\ldots,a_{2L+1})=f(a_{\pi(1)},\ldots,\allowbreak a_{\pi(2L+1)})$ for every
$a_1,\ldots,a_{2L+1}\in A$ and every permutation $\pi:[2L+1]\to[2L+1]$ that preserves
parity, and $f(a_1,\ldots,a_{2L-1},a,a)=f(a_1,\ldots,a_{2L-1},a',a')$ for every
$a_1,\ldots,a_{2L-1},a,a'\in A$. Intuitively, an alternating operation is
invariant under permutations of its odd and even coordinates and has the
property that adjacent coordinates cancel each other out.

The power of $\AIP$ for PCSPs is characterised in the following result.

\begin{thm}[\cite{BBKO21}]\label{thm:aip}
Let $(\A,\B)$ be a PCSP template. The following are equivalent:
  \begin{enumerate}
    \item[(1)] $\AIP$ solves $\PCSP(\A,\B)$.
    \item[(2)] $\Pol(\A,\B)$ admits a minion homomorphism from $\Zaff$.
    \item[(3)] $\Pol(\A,\B)$ contains alternating operations of all odd arities.
  \end{enumerate}
\end{thm}

\subsection{BLP+AIP}\label{subsec:blp-aip}
The \emph{combined basic LP and affine IP algorithm} ($\BLP+\AIP$) 
is presented in Algorithm~\ref{alg:blp-aip}.

\begin{algorithm}[tbh] 
	\SetAlgoLined
  \KwIn{\quad\ \ \ an instance $\X$ of $\PCSP(\A,\B)$ of signature $\sigma$}
  \KwOut{\quad \textsc{yes} if $\X\to\A$ and \textsc{no} if $\X\not\to\B$}
  \medskip
  find a relative interior point $(\lambda_{\bx,R}(\ba))_{R\in\sigma, \bx\in R^\X,\ba\in R^\A}$ of $\BLP(\X,\A)$\;
  \If{no relative interior point exists}{return \textsc{no}\;}
  refine $\AIP(\X,\A)$ by setting $\tau_{\bx,R}(\ba)=0$ if $\lambda_{\bx,R}(\ba)=0$\;
  \If{the refined $\AIP(\X,\A)$ accepts}{return \textsc{yes}\;}{return \textsc{no}\;}
\caption{The $\BLP+\AIP$ algorithm} \label{alg:blp-aip}
\end{algorithm}

If $\X\to\A$ then $\BLP+\AIP$ accepts $\X$~\cite{bgwz20}. We say that $\BLP+\AIP$
\emph{solves} $\PCSP(\A,\B)$ if for every instance $\X$ accepted by $\BLP+\AIP$
we have $\X\to\B$. 

We denote by $\Mblpaip$
the minion whose $L$-ary objects are
pairs $(\bq,\bz)$, where $\bq\in\Q^L$ is a stochastic vector and $\bz\in\Z^L$ is
an affine vector, with the property that, for every $i\in [L]$, $q_i=0$
implies $z_i=0$. As before, the minor operation is defined as in
Section~\ref{sec:minion}; i.e., if $(\bq,\bz)\in\Mblpaip^{(L)}$
and $\pi:[L]\to[L']$, then $(\bq,\bz)_{/\pi}=({P_\pi}\bq,{P_\pi}\bz)$, where ${P_\pi}$
is the $L'\times L$ matrix whose $(i,j)$-th entry is $1$ if $\pi(j)=i$, and
$0$ otherwise.

A $(2L+1)$-ary operation $f:A^{2L+1}\to B$ is called \emph{$2$-block symmetric}
if $f(a_1,\ldots,a_{2L+1})=f(a_{\pi(1)},\ldots,a_{\pi(2L+1)})$ for every
$a_1,\ldots,a_{2L+1}\in A$ and every permutation $\pi:[2L+1]\to[2L+1]$ that
preserves parity.

The power of $\BLP+\AIP$ for PCSPs is characterised in the following result.

\begin{thm}[\cite{bgwz20}]\label{thm:blp-aip}
Let $(\A,\B)$ be a PCSP template. The following are equivalent:
  \begin{enumerate}
    \item[(1)] $\BLP+\AIP$ solves $\PCSP(\A,\B)$.
    \item[(2)] $\Pol(\A,\B)$ admits a minion homomorphism from $\Mblpaip$.
    \item[(3)] $\Pol(\A,\B)$ contains $2$-block-symmetric operations of all odd arities.
  \end{enumerate}
\end{thm}

\section{Proof of Lemma~\ref{lem:free}} \label{sec:free-proof}

In this section, we shall prove Lemma~\ref{lem:free}, which we restate below.
The proof is based on that of~\cite[Lemma~4.4]{BBKO21},
which 
concerns minions of functions.

\begin{lem*}[Lemma~\ref{lem:free} restated]
Let $\mathscr{M}$ be a minion and let $(\A,\B)$ be a PCSP template. Then there is a minion homomorphism from $\mathscr{M}$ to $\Pol(\A,\B)$ if and only if $\mathbb{F}_\mathscr{M}(\A)\to\B$.
\end{lem*}

\begin{proof}
Let $A=[n]$, and let $\sigma$ be the signature of $\A$ and $\B$. Suppose
  $\xi:\mathscr{M}\to\Pol(\A,\B)$ is a minion homomorphism, and consider the function
\begin{align*}
f:\mathscr{M}^{(n)}&\to B\\
M&\mapsto \xi(M)(1,\dots,n).
\end{align*}
For $R\in\sigma$ of arity $k$, consider a tuple $(M_1,\dots,M_k)\in R^{\mathbb{F}_\mathscr{M}(\A)}$. List the elements of $R^\A$ as $\ba^{(1)},\dots,\ba^{(m)}$. From Definition~\ref{def:free}, $\exists Q\in \mathscr{M}^{(m)}$ such that $M_i=Q_{/\pi_i}$ for each $i\in [k]$, where $\pi_i:[m]\to A$ maps $j$ to the $i$-th coordinate of $\ba^{(j)}$. It follows that, for each $i\in [k]$,
\begin{align*}
f(M_i)&=f(Q_{/\pi_i})=\xi(Q_{/\pi_i})(1,\dots,n)=\xi(Q)_{/\pi_i}(1,\dots,n)=\xi(Q)(\pi_i(1),\dots,\pi_i(m)).
\end{align*}
Hence,
\begin{align*}
f(M_1,\dots,M_k)&=(\xi(Q)(\pi_1(1),\dots,\pi_1(m)), \dots,\xi(Q)(\pi_k(1),\dots,\pi_k(m)))=\xi(Q)(\ba^{(1)},\dots,\ba^{(m)})\in R^\B
\end{align*}
  since $\xi(Q)$ is a polymorphism of $(\A,\B)$. Therefore, $f$ is a homomorphism from $\mathbb{F}_\mathscr{M}(\A)$ to $\B$.

Conversely, let $f:\mathbb{F}_\mathscr{M}(\A)\to\B$ be a homomorphism, and consider the function $\xi:\mathscr{M}\to\Pol(\A,\B)$ defined by $\xi(M)(a_1,\dots,a_L)=f(M_{/\rho})$ for each $L\in\N$, $M\in\mathscr{M}^{(L)}$, $(a_1,\dots,a_L)\in A^L$, where
\begin{align*}
\rho:[L]&\to[n]\\
i&\mapsto a_i.
\end{align*}
Let us first check that $\xi$ is well defined -- i.e., that $\xi(M)\in\Pol(\A,\B)$. For $R\in\sigma$ of arity $k$, consider a matrix $Z\in A^{L,k}$ such that each row of $Z$ corresponds to a tuple in $R^\A$. We need to show that $\xi(M)(Z)\in R^\B$. Consider the maps
\begin{align*}
\tau:[L]&\to R^\A & \rho_j:[L]&\to[n]& \pi_j:R^\A&\to [n]\\
i&\mapsto Z^T\be_i,& i&\mapsto \be_i^TZ\be_j, & \ba&\mapsto \be_j^T\ba,
\end{align*}
for $j\in[k]$. Observe that $\rho_j=\pi_j\circ\tau$, and set $Q=M_{/\tau}\in\mathscr{M}^{(|R^\A|)}$. We obtain
\begin{align*}
\xi(M)(Z)&=f(M_{/\rho_1},\dots,M_{/\rho_k})=
f(M_{/\pi_1\circ\tau},\dots,M_{/\pi_k\circ\tau})
=
f(Q_{/\pi_1},\dots,Q_{/\pi_k})\in R^\B
\end{align*}
since $(Q_{/\pi_1},\dots,Q_{/\pi_k})\in R^{\mathbb{F}_\mathscr{M}(\A)}$ and $f$ is a homomorphism.
Finally, we show that $\xi$ is a minion homomorphism. Clearly, $\xi$ preserves arities. To check that it preserves minors, let $M\in\mathscr{M}^{(L)}$ and take a map $\pi:[L]\to [L']$. Given $(a_1,\dots,a_{L'})\in A^{L'}$, consider the maps
\begin{align*}
\rho':[L']&\to [n]&\rho'':[L]&\to [n]\\
i&\mapsto a_i, & i&\mapsto a_{\pi(i)},
\end{align*}
and observe that $\rho''=\rho'\circ\pi$. We obtain
\begin{align*}
\xi(M_{/\pi})(a_1,\dots,a_{L'})&=f((M_{/\pi})_{/\rho'})=f(M_{/\rho'\circ\pi})=f(M_{/\rho''})=\xi(M)(a_{\pi(1)},\dots,a_{\pi(L)})\\
&=\xi(M)_{/\pi}(a_1,\dots,a_{L'}),
\end{align*}
which yields $\xi(M_{/\pi})=\xi(M)_{/\pi}$, as desired.
\end{proof}

\section{H-block-symmetric polymorphisms} \label{sec:block_sym_pol}
Let $\mathcal{C}=(\mathcal{C}_1,\dots,\mathcal{C}_\ell)$ be a partition of ${c}\in \N$; i.e., the sets $\mathcal{C}_i$ are pairwise disjoint and their union is $[{c}]$. Let ${c}_i=|\mathcal{C}_i|$, so that ${c}=\sum_{i\in[\ell]}{c}_i$. For each $i\in[\ell]$, we consider the unique monotonically increasing function $\vartheta_i:[{c}_i]\rightarrow[{c}]$ such that $\vartheta_i([{c}_i])=\mathcal{C}_i$. We also consider the function $\chi_i:\mathcal{C}_i\to[{c}_i]$ such that $\vartheta_i\circ \chi_i$ is the inclusion map of $\mathcal{C}_i$ in $[{c}]$. Given ${c}'\in\N$ and a map $\pi:[{c}]\rightarrow[{c}']$, we let $\pi_{(i)}=\pi\circ \vartheta_i$.

\begin{defn}
Let $A,B$ be finite sets, and consider a function $f:A^{c}\rightarrow B$ for
  some ${c}\in\N$. Given an $m\times |A|$ tie matrix $H$ and a partition
  $\mathcal{C}=(\mathcal{C}_1,\dots,\mathcal{C}_\ell)$ of ${c}$, we say that $f$
  is \emph{$H$-\,$\mathcal{C}$-block-symmetric} if
\begin{align*}
f_{/\pi}(\ba)=f(\ba) &&& \forall \pi: [{c}]\rightarrow [{c}] \mbox{ permutation such that }\pi(\mathcal{C}_i)=\mathcal{C}_i \hspace{.2cm}\forall i\in [\ell],\\
&&&\forall \ba\in A^{c}
\mbox{ such that }({P^T_{\vartheta_i}}\ba)^\# \mbox{ is $H$-tieless}\hspace{.2cm}\forall i\in [\ell].
\end{align*}
\end{defn}
\noindent We say that $f$ is \emph{$H$-block-symmetric} with \emph{width} $W$ if $W$ is
the largest integer for which there is a partition $\mathcal{C}$ of ${c}$ such
that each part of $\mathcal{C}$ has size at least $W$ and $f$ is
\emph{$H$-\,$\mathcal{C}$-block-symmetric}.\footnote{The notion of
$H$-block-symmetric operation is the $H$-analogue of that of block-symmetric operation in~\cite{bgwz20} (cf.~Theorem~\ref{thm:blp-aip}).} Without loss of generality, we consider $A=[n]$.
\begin{lem}
\label{lem_H_tieless_blocks}
Let $f:A^{c}\rightarrow B$ be $H$-\,$\mathcal{C}$-block-symmetric for some $m\times n$ tie matrix $H$ ($m\in\N$) and some partition $\mathcal{C}=(\mathcal{C}_1,\dots,\mathcal{C}_\ell)$ of ${c}$. Consider two maps $\pi,\tilde{\pi}:[{c}]\rightarrow [n]$ such that, for each $i\in [\ell]$, ${P_{\pi_{(i)}}}\bone_{{c}_i}={P_{\tilde{\pi}_{(i)}}}\bone_{{c}_i}$ and the vector ${P_{\pi_{(i)}}}\bone_{{c}_i}$ is $H$-tieless. Then
\begin{align*}
f_{/\pi}(1,\dots,n)=f_{/\tilde{\pi}}(1,\dots,n).
\end{align*}
\end{lem}
\begin{proof}
For $i\in [\ell]$ and $a\in [n]$, we have
\begin{align*}
|\pi_{(i)}^{-1}(a)|=\be_a^T{P_{\pi_{(i)}}}\bone_{{c}_i}
=
\be_a^T{P_{\tilde\pi_{(i)}}}\bone_{{c}_i}
=
|\tilde\pi_{(i)}^{-1}(a)|.
\end{align*}
Hence, we can consider bijections $\varphi_{i,a}:\pi_{(i)}^{-1}(a)\rightarrow\tilde{\pi}_{(i)}^{-1}(a)$ for each $i\in[\ell],a\in [n]$. The union 
\begin{align*}
\varphi_i=\bigcup_{a\in [n]}\varphi_{i,a}:[{c}_i]\rightarrow [{c}_i]
\end{align*}
is also a bijection. Define $\varphi:[{c}]\to[{c}]$ by letting $\varphi\big{|}_{\mathcal{C}_i}=\vartheta_i\circ\varphi_i\circ\chi_i$ for each $i\in [\ell]$. Notice that $\varphi(\mathcal{C}_i)=\mathcal{C}_i$ for each $i\in[\ell]$, so $\varphi$ is a bijection. Take $j\in [{c}]$ and suppose that $j\in\mathcal{C}_i$. We have
\begin{align*}
(\tilde{\pi}\circ \varphi)(j)=\tilde{\pi}(\varphi(j))=\tilde{\pi}(\vartheta_i(\varphi_i(\chi_i(j))))=
\tilde{\pi}_{(i)}(\varphi_{i,\pi_{(i)}(\chi_i(j))}(\chi_i(j)))=\pi_{(i)}(\chi_i(j))=(\pi\circ\vartheta_i\circ\chi_i)(j)=\pi(j)
\end{align*}
and, hence, $\tilde{\pi}\circ\varphi=\pi$.
Let $\tilde{\ba}=(\tilde{\pi}(1),\dots,\tilde{\pi}({c}))$. Notice that, for each $i\in[\ell]$ and $a\in [n]$,
\begin{align*}
\be_a^T({P^T_{\vartheta_i}}\tilde{\ba})^\#
&=
|\{j\in[{c}_i]:\be_j^T{P^T_{\vartheta_i}}\tilde{\ba}=a \}|
=
|\{j\in[{c}_i]:\be_{\vartheta_i(j)}^T\tilde{\ba}=a \}|
=
|\{j\in[{c}_i]:\tilde{\pi}(\vartheta_i(j))=a \}|\\
&=
|\{j\in[{c}_i]:\tilde{\pi}_{(i)}(j)=a \}|
=
\be_a^T{P_{\tilde{\pi}_{(i)}}}\bone_{{c}_i}
\end{align*}
and, therefore, $({P^T_{\vartheta_i}}\tilde{\ba})^\#={P_{\tilde{\pi}_{(i)}}}\bone_{{c}_i}={P_{{\pi}_{(i)}}}\bone_{{c}_i}$, which is $H$-tieless.
Using that $f$ is $H$-\,$\mathcal{C}$-block-symmetric, we find
\begin{align*}
f_{/\tilde{\pi}}(1,\dots,n)=f(\tilde{\ba})=f_{/\varphi}(\tilde{\ba})=(f_{/\varphi})_{/\tilde{\pi}}(1,\dots,n)=
f_{/\tilde{\pi}\circ\varphi}(1,\dots,n)=f_{/\pi}(1,\dots,n),
\end{align*}
as required.
\end{proof}
\begin{thm}
\label{thm_CLAP_works_H_block_sym}
Let $(\A,\B)$ be a PCSP template and suppose $\Pol(\A,\B)$ contains
  $H$-block-symmetric operations of arbitrarily large width for some $m\times |A|$ tie
  matrix $H$, $m\in\N$. Then there exists a minion homomorphism from
  $\ourMinion$ to $\Pol(\A,\B)$. 
\end{thm}
\begin{proof}
For $D\in\N$, consider the subminion $\ourMinion_D$ of $\ourMinion$ described in Section~\ref{sec_compactness_argument}. Observe that $S=\{M:(M,\bmu)\in\ourMinion_D^{(n)}\}$ is a finite set of skeletal matrices. Therefore, we can apply the Tiebreak Lemma~\ref{lemma_tie_terminator} to find a stochastic finitely supported vector ${\bv}\in \Q^{\aleph_0}$ with $\be_1^T{\bv}>0$ such that $M{\bv}$ is $H$-tieless for any $M\in S$. Since ${\bv}$ is finitely supported, we can find $N'\in\N$ such that $N'{\bv}$ has integer entries. Let $\sigma_1^H$ denote the largest singular value of $H$ -- i.e., the square root of the largest eigenvalue of $H^TH$. Set $N=2\lceil\sigma_1^H+1\rceil D^2N'$, and let $f$ be an $H$-block-symmetric polymorphism of width ${W}\geq N^2$. Letting $c$ be the arity of $f$, consider a partition $\mathcal{C}=(\mathcal{C}_1,\dots,\mathcal{C}_\ell)$ of $c$ such that $c_i=|\mathcal{C}_i|\geq W$ for each $i\in [\ell]$ and $f$ is $H$-\,$\mathcal{C}$-block-symmetric. Write ${c}_i=N\alpha_i+\beta_i$ with $\alpha_i,\beta_i\in \N_0$, $\beta_i\leq N-1$. Note that $N^2\leq W\leq {c}_i= N\alpha_i+\beta_i\leq N\alpha_i+N-1<N(\alpha_i+1)$, so $N<\alpha_i+1$ and, hence, $\beta_i<\alpha_i$.

Consider the function 
\begin{align*}
\xi_D:\ourMinion_D\rightarrow\Pol(\A,\B)
\end{align*} 
defined as follows. Given $L\in\N$ and $(M,\bmu)\in\ourMinion_D^{(L)}$, for each $i\in [\ell]$ take the map $\varphi_i:[{c}_i]\rightarrow [L]$ such that the corresponding $L\times {c}_i$ matrix $P_{\varphi_i}$ is
\begin{align}
\label{expr_P_phi_block}
{P_{\varphi_i}}=
\begin{pmatrix}
\bone^T_{\be_1^T(\alpha_iNM{\bv}+\beta_i\bmu)} & \bzero^T & \dots & \bzero^T \\ 
\bzero^T & \bone^T_{\be_2^T(\alpha_iNM{\bv}+\beta_i\bmu)} & \dots & \bzero^T \\ 
\vdots & \vdots & \ddots & \vdots \\ 
\bzero^T & \bzero^T & \dots & \bone^T_{\be_L^T(\alpha_iNM{\bv}+\beta_i\bmu)}
\end{pmatrix}.
\end{align}
To verify that~\eqref{expr_P_phi_block} is well defined, observe first that
\begin{align*}
\sum_{j=1}^L \be_j^T(\alpha_iNM{\bv}+\beta_i\bmu)=\bone_L^T(\alpha_iNM{\bv}+\beta_i\bmu)=\alpha_iN\bone_L^TM{\bv}+\beta_i\bone_L^T\bmu=\alpha_iN\bone_{\aleph_0}^T{\bv}+\beta_i=\alpha_iN+\beta_i={c}_i.
\end{align*}
Moreover, for each $j\in [L]$, $\be_j^T(\alpha_iNM{\bv}+\beta_i\bmu)=\be_j^T(2\alpha_i\lceil\sigma_1^H+1\rceil D(DM)(N'{\bv})+\beta_i\bmu)$ is an integer. If $\be_j^T(\alpha_iNM{\bv}+\beta_i\bmu)$ was negative, then $\be_j^T\bmu<0$. By the requirement $(c_4)$ in Definition~\ref{def:minion}, this would imply that $\be_j^TM\be_1>0$ and, hence, $0<\be_j^TM\be_1\be_1^T{\bv}\leq \be_j^TM{\bv}$. As a consequence, $\be_j^T(DM)(N'{\bv})\geq 1$ so that
\begin{align*}
\be_j^T(\alpha_iNM{\bv}+\beta_i\bmu)=2\alpha_i\lceil\sigma_1^H+1\rceil D\be_j^T(DM)(N'{\bv})+\beta_i\be_j^T\bmu\geq 2\alpha_i\lceil\sigma_1^H+1\rceil D+\beta_i\be_j^T\bmu\geq \alpha_iD-\beta_iD>0,
\end{align*}
which is a contradiction. In conclusion, the numbers $\be_j^T(\alpha_iNM{\bv}+\beta_i\bmu)$ are nonnegative integers summing up to ${c}_i$, so~\eqref{expr_P_phi_block} is well defined.

Consider the function $\varphi:[c]\to [L]$ defined by $\varphi\big{|}_{\mathcal{C}_i}=\varphi_i\circ\chi_i$ $\forall i\in [\ell]$, and let $\xi_D((M,\bmu))\coloneqq f_{/\varphi}$. Clearly, $\xi_D((M,\bmu))\in \Pol(\A,\B)$. We claim that the map $\xi_D$ is a minion homomorphism. It is straightforward to check that $\xi_D$ preserves arities so, to conclude, we need to show that it also preserves minors. Take $L'\in\N$ and choose a map $\pi:[L]\rightarrow [L']$. Letting $\tilde{\varphi}_i:[{c}_i]\rightarrow [L']$ be the map corresponding to the matrix
\begin{align*}
{P_{\tilde\varphi_i}}=
\begin{pmatrix}
\bone^T_{\be_1^T(\alpha_iN{P_\pi}M{\bv}+\beta_i{P_\pi}\bmu)} & \bzero^T & \dots & \bzero^T \\ 
\bzero^T & \bone^T_{\be_2^T(\alpha_iN{P_\pi}M{\bv}+\beta_i{P_\pi}\bmu)} & \dots & \bzero^T \\ 
\vdots & \vdots & \ddots & \vdots \\ 
\bzero^T & \bzero^T & \dots & \bone^T_{\be_{L'}^T(\alpha_iN{P_\pi}M{\bv}+\beta_i{P_\pi}\bmu)}
\end{pmatrix}
\end{align*}
for each $i\in[\ell]$, and considering $\tilde{\varphi}:[c]\to [L']$ such that $\tilde{\varphi}\big{|}_{\mathcal{C}_i}=\tilde{\varphi}_i\circ\chi_i$ $\forall i\in[\ell]$, we see that $\xi_D((M,\bmu)_{/\pi})=f_{/\tilde{\varphi}}$. Moreover, $\xi_D((M,\bmu))_{/\pi}=(f_{/\varphi})_{/\pi}=f_{/\pi\circ\varphi}$, where $\varphi$ is the map defined above. Take ${\ba}=(a_1,\dots,a_{L'})\in A^{L'}$, and consider the map 
\begin{align*}
\pi_{\ba}:[L']&\rightarrow [n]\\
i&\mapsto a_i.
\end{align*}
Observe that 
\begin{align}
\label{eqn_f_phi_phi_tilde_block}
f_{/\tilde{\varphi}}(\ba)&=&(f_{/\tilde{\varphi}})_{/\pi_{\ba}}(1,\dots,n)&=&f_{/\pi_{\ba}\circ \tilde{\varphi}}(1,\dots,n)&&\mbox{and, similarly,}\notag\\
f_{/\pi\circ\varphi}(\ba)&=&(f_{/\pi\circ\varphi})_{/\pi_{\ba}}(1,\dots,n)&=&f_{/\pi_{\ba}\circ\pi\circ\varphi}(1,\dots,n).
\end{align}
Notice that, for each $i\in[\ell]$, $\varphi\circ\vartheta_i=\varphi_i$ and $\tilde\varphi\circ\vartheta_i=\tilde\varphi_i$. Hence,
\begin{align*}
{P_{(\pi_{\ba}\circ \tilde{\varphi})_{(i)}}}\bone_{c_i}
&=
{P_{\pi_{\ba}\circ \tilde{\varphi}\circ\vartheta_i}}\bone_{c_i}
=
{P_{\pi_{\ba}}}{P_{\tilde{\varphi}\circ\vartheta_i}}\bone_{c_i}
=
{P_{\pi_{\ba}}}{P_{\tilde\varphi_i}}\bone_{c_i}
=
{P_{\pi_{\ba}}}(\alpha_iN{P_\pi}M{\bv}+\beta_i{P_\pi}\bmu)\\
&=
{P_{\pi_{\ba}}}{P_\pi}(\alpha_iNM{\bv}+\beta_i\bmu)=
{P_{\pi_{\ba}}}{P_\pi}{P_{\varphi_i}}\bone_{c_i}=
{P_{\pi_{\ba}\circ\pi\circ\varphi\circ\vartheta_i}}\bone_{c_i}=
{P_{(\pi_{\ba}\circ\pi\circ\varphi)_{(i)}}}\bone_{c_i}.
\end{align*}
We claim that the vector ${P_{(\pi_{\ba}\circ \tilde{\varphi})_{(i)}}}\bone_{c_i}={P_{\pi_{\ba}\circ \tilde{\varphi}_i}}\bone_{c_i}$ is $H$-tieless. Let ${\bu}=(u_t)=H{P_{\pi_{\ba}\circ \tilde{\varphi}_i}}\bone_{c_i}$; the claim is equivalent to ${\bu}$ being tieless. Let ${\bw}=(w_t)=\alpha_iNH{P_{\pi_{\ba}\circ\pi}}M{\bv}$ and ${\bz}=(z_t)=\beta_iH{P_{\pi_{\ba}\circ\pi}}\bmu$, so that ${\bu}={\bw}+{\bz}$. Choose ${t},{t}'\in [m]$ such that ${t}\neq {t}'$ and $u_{t}\neq 0$. We need to show that $u_{t}\neq u_{{t}'}$. Suppose $w_{t}=0$. We can write $H^T\be_{t}=\sum_{g\in G}\lambda_g \be_g$ for $G=\supp(H^T\be_t)$, where each $\lambda_g$ is a positive integer (note that $G\neq\emptyset$ since, otherwise, $H^T\be_t=\bzero_n$, which would imply $u_t=0$). Let ${F}=(\pi_{\ba}\circ \pi)^{-1}(G)$. From $w_{t}=0$, we obtain
\begin{align*}
&&0=\be_{t}^TH{P_{\pi_{\ba}\circ\pi}}M{\bv}=(H^T\be_{t})^T{P_{\pi_{\ba}\circ\pi}}M{\bv}
&=&
\sum_{g\in G}\lambda_g\be_g^T{P_{\pi_{\ba}\circ\pi}}M{\bv}
=
\sum_{g\in G}\lambda_g\sum_{j\in (\pi_{\ba}\circ\pi)^{-1}(g)}\be_j^TM{\bv}
\end{align*}
and, hence, the following chain of implications holds:
\begin{align*}
&&\begin{array}{lllll}
&\displaystyle
0=\sum_{g\in G}\sum_{j\in (\pi_{\ba}\circ\pi)^{-1}(g)}\be_j^TM{\bv}
=
\sum_{j\in {F}}\be_j^TM{\bv}
&\displaystyle
\hspace{.2cm}\Rightarrow\hspace{.2cm}
&\displaystyle
\be_j^TM{\bv}=0&{\forall j\in {F}}
\\[7pt]
\Rightarrow\hspace{.2cm}
&\displaystyle
\be_j^TM\be_1=0\hspace{1cm}{\forall j\in {F}}
&\hspace{.2cm}\Rightarrow\hspace{.2cm}
&\displaystyle
\be_j^T\bmu=0&{\forall j\in {F}}
\end{array}
\end{align*}
(where the second implication follows from $\be_1^T{\bv}>0$, and the third follows from $(c_4)$ in Definition~\ref{def:minion}). Hence,
\begin{align*}
z_{t}=\beta_i\be_{t}^TH{P_{\pi_{\ba}\circ\pi}}\bmu
=
\beta_i\sum_{g\in G}\lambda_g\be_g^T{P_{\pi_{\ba}\circ\pi}}\bmu
=
\beta_i\sum_{g\in G}\lambda_g\sum_{j\in (\pi_{\ba}\circ\pi)^{-1}(g)}\be_j^T\bmu=0,
\end{align*}
so that $u_{t}=w_{t}+z_{t}=0$, a contradiction. Hence, $w_{t}>0$. Observe that $(M_{/\pi_{\ba}\circ\pi},\bmu_{/\pi_{\ba}\circ\pi})\in \ourMinion_D^{(n)}$ and, hence, $M_{/\pi_{\ba}\circ\pi}\in S$. By the choice of ${\bv}$, this implies that the vector ${P_{\pi_{\ba}\circ \pi}}M{\bv}=M_{/\pi_{\ba}\circ \pi}{\bv}$ is $H$-tieless; i.e., $H{P_{\pi_{\ba}\circ \pi}}M{\bv}$ is tieless. It follows that the vector $H{P_{\pi_{\ba}\circ \pi}}(DM)(N'{\bv})=\frac{1}{2\alpha_i\lceil\sigma_1^H+1\rceil D}{\bw}$ is also tieless; being it entrywise integer, and since $\frac{1}{2\alpha_i\lceil\sigma_1^H+1\rceil D}w_{t}>0$, we obtain
\begin{align*}
\left|\frac{1}{2\alpha_i\lceil\sigma_1^H+1\rceil D}w_{t}-\frac{1}{2\alpha_i\lceil\sigma_1^H+1\rceil D}w_{{t}'}\right|\geq 1&&\mbox{that yields}&&|w_{t}-w_{{t}'}|\geq 2\alpha_i\lceil\sigma_1^H+1\rceil D.
\end{align*}
Denote the $\ell_1$-norm\ and the $\ell_2$-norm of a vector by $\|\cdot\|_1$ and
$\|\cdot\|_2$, respectively. Recall that the largest singular value of a matrix is its spectral operator norm -- i.e., $\sigma_1^H=\max_{\bzero\neq\bx\in \R^{n}}\frac{\|H\bx\|_2}{\|\bx\|_2}$ (see~\cite{Horn2012matrix}). In particular, $\|H\bx\|_2\leq\sigma_1^H\|\bx\|_2$ for each vector $\bx$ of size $n$.
Using the Cauchy-Schwarz inequality and the fact that the $\ell_1$-norm of a vector is greater than or equal to its $\ell_2$-norm, we find
\begin{align*}
|z_{t}-z_{{t}'}|&=\beta_i|(\be_{t}-\be_{{t}'})^TH{P_{\pi_{\ba}\circ\pi}}\bmu|\leq 
\beta_i\|\be_{t}-\be_{{t}'}\|_2\|H{P_{\pi_{\ba}\circ\pi}}\bmu\|_2
\leq
\beta_i\|\be_{t}-\be_{{t}'}\|_2\sigma_1^H\|{P_{\pi_{\ba}\circ\pi}}\bmu\|_2\\
&\leq 
\beta_i\|\be_{t}-\be_{{t}'}\|_1\lceil\sigma_1^H+1\rceil\|{P_{\pi_{\ba}\circ\pi}}\bmu\|_1
=
2\beta_i\lceil\sigma_1^H+1\rceil \bone_n^T|{P_{\pi_{\ba}\circ\pi}}\bmu|
\leq
2\beta_i\lceil\sigma_1^H+1\rceil \bone_n^T{P_{\pi_{\ba}\circ\pi}}|\bmu|\\
&=2\beta_i\lceil\sigma_1^H+1\rceil \bone_L^T|\bmu|
\leq 
2\beta_i\lceil\sigma_1^H+1\rceil D < 2\alpha_i\lceil\sigma_1^H+1\rceil D.
\end{align*}
We conclude the proof of the claim by noting that 
\begin{align*}
|u_{t}-u_{{t}'}|=|(w_{t}-w_{{t}'})-(z_{{t}'}-z_{t})|\geq |w_{t}-w_{{t}'}|-|z_{t}-z_{{t}'}|>2\alpha_i\lceil\sigma_1^H+1\rceil D-2\alpha_i \lceil\sigma_1^H+1\rceil D=0,
\end{align*}
which implies $u_{t}\neq u_{{t}'}$. As a consequence, the vector ${P_{(\pi_{\ba}\circ \tilde{\varphi})_{(i)}}}\bone_{c_i}$ is $H$-tieless. We can then apply Lemma~\ref{lem_H_tieless_blocks} to conclude that $f_{/\pi_{\ba}\circ
\tilde{\varphi}}(1,\dots,n)=f_{/\pi_{\ba}\circ\pi\circ\varphi}(1,\dots,n)$. Hence, by~\eqref{eqn_f_phi_phi_tilde_block}, $f_{/\tilde{\varphi}}=f_{/\pi\circ\varphi}$. Therefore,
$\xi_D((M,\bmu)_{/\pi})=\xi_D((M,\bmu))_{/\pi}$, as required. It follows that $\xi_D$ is a minion homomorphism. 

Since the set of polymorphisms of $(\A,\B)$ of arity $L$ is finite for every $L\in\N$, we can apply Proposition~\ref{prop_compactness_argument} to conclude that there exists a minion homomorphism $\zeta:\ourMinion\rightarrow\Pol(\A,\B)$.
\end{proof}

{\small
\bibliographystyle{plainurl}
\bibliography{cz}
}

\end{document}